\documentclass[conference]{IEEEtran}
\IEEEoverridecommandlockouts
\usepackage{cite}
\usepackage{amsmath,amsthm,amssymb,amsfonts,braket}
\usepackage{multirow}
\usepackage{algorithmic}
\usepackage{breqn}
\usepackage[qm]{qcircuit}
\usepackage{textcomp}
\usepackage{xcolor}
\usepackage{enumitem}
\usepackage{dirtytalk}
\usepackage[export]{adjustbox}
\usepackage{braket}
\usepackage{graphicx}
\usepackage{subcaption}
\usepackage{tabularx,booktabs}
\newtheorem{theorem}{Theorem}
\newtheorem{lemma}{Lemma}
\usepackage{tablefootnote}

\def\BibTeX{{\rm B\kern-.05em{\sc i\kern-.025em b}\kern-.08em
    T\kern-.1667em\lower.7ex\hbox{E}\kern-.125emX}}

\PassOptionsToPackage{hyphens}{url}    
\usepackage{tikz,xcolor,hyperref}
\usepackage{microtype}

    \definecolor{lime}{HTML}{A6CE39}
\DeclareRobustCommand{\orcidicon}{
	\begin{tikzpicture}
	\draw[lime, fill=lime] (0,0) 
	circle [radius=0.16] 
	node[white] {{\fontfamily{qag}\selectfont \tiny ID}};
	\draw[white, fill=white] (-0.0625,0.095) 
	circle [radius=0.007];
	\end{tikzpicture}
	\hspace{-2mm}
}
\foreach \x in {A, ..., Z}{
	\expandafter\xdef\csname orcid\x\endcsname{\noexpand\href{https://orcid.org/\csname orcidauthor\x\endcsname}{\noexpand\orcidicon}}
}

\begin{document}

\title{ A Modified Depolarization Approach for Efficient Quantum Machine Learning}

\author{\IEEEauthorblockN{Bikram Khanal\orcidA{}}
\IEEEauthorblockA{\textit{Department of Computer Science} \\
\textit{Baylor University}\\
Waco, TX \\
bikram\_khanal1@baylor.edu}
\and
\IEEEauthorblockN{Pablo Rivas\orcidB{}, \emph{Senior, IEEE}}
\IEEEauthorblockA{\textit{Department of Computer Science} \\
\textit{Baylor University}\\
Waco, TX \\
pablo\_rivas@baylor.edu}}
\maketitle

\begin{abstract}
Quantum Computing in the Noisy Intermediate-Scale Quantum (NISQ) era has shown promising applications in machine learning, optimization, and cryptography. Despite the progress, challenges persist due to system noise, errors, and decoherence that complicate the simulation of quantum systems. The depolarization channel is a standard tool for simulating a quantum system's noise. However, modeling such noise for practical applications is computationally expensive when we have limited hardware resources, as is the case in the NISQ era. We propose a modified representation for a single-qubit depolarization channel with two Kraus operators based only on $X$ and $Z$ Pauli matrices. Our approach reduces the computational complexity from six to four matrix multiplications per execution of a channel. Experiments on a Quantum Machine Learning (QML) model on the Iris dataset across various circuit depths and depolarization rates validate that our approach maintains the model's accuracy while improving efficiency. This simplified noise model enables more scalable simulations of quantum circuits under depolarization, advancing capabilities in the NISQ era.
\end{abstract}

\section{Introduction}
Quantum Computing has seen significant progress in recent years, with the development of quantum algorithms for a variety of applications, including machine learning~\cite{mitarai2018quantum,havlivcek2019supervised,liu2021rigorous,sajjan2021quantum,cai2015entanglement,ciliberto2018quantum}, optimization~\cite{farhi2014quantum,mcclean2016theory,Rebentrost2016Quantum,Bittel2021Training,Rebentrost2018Quantum}, and cryptography~\cite{Broadbent2015Quantum,Padamvathi2016Quantum,Lai2017Fast,Pirandola2019Advances}. However, the development of quantum algorithms is still in its infancy, and many of the algorithms that have been developed are not yet ready for practical use~\cite{Harrow2017Quantum, preskill2018quantum}. Due to the susceptibility of NISQ device operations to errors and decoherence~\cite{pop00002}, simulating quantum systems remains a major challenge in developing quantum algorithms~\cite{preskill2018quantum}.

In the NISQ era, system noise is not merely a nuisance to be minimized but a fundamental force shaping the field of QML research.  Interestingly, a considerable number of works have chosen to regard noise not as a challenge but as an opportunity to advance their research. Studies have shown that quantum learning of n-bit parity functions remains remarkably efficient under depolarizing noise, a testament to the inherent resilience of quantum algorithms compared to their classical counterparts~\cite{Cross2014Quantum}. This early work demonstrated the potential for quantum algorithms to maintain a learning advantage even in noisy conditions. While traditionally viewed as a detrimental factor to quantum computation, depolarization noise under certain conditions can enhance the robustness and functionality of quantum learning algorithms against adversarial attacks~\cite{pop00001, pop00023,west2023towards, lu2020quantum,pop00008}. This counterintuitive finding highlights the potential of noise to endow quantum models with robustness against malicious attempts to manipulate the model's outputs.

However, harnessing the power of noise as a training tool requires careful consideration. The effectiveness of adversarial training techniques, for example, hinges on the assumption that the test attack and the training attack employ the same methods to generate adversarial examples. In real-world scenarios where attackers may employ diverse and unknown strategies, this advantage is not guaranteed~\cite{bai2021recent, kang2019transfer}. Therefore, deriving robust guarantees against worst-case scenarios remains crucial for building truly secure and resilient quantum learning algorithms.

The challenges posed by noise extend beyond algorithm design, impacting the very foundations of QML. The inherent noise in NISQ machines also presents significant challenges to the learning capabilities of Quantum Neural Networks~\cite{pop00002}. The presence of system noise can significantly diminish the quantum kernel advantage~\cite{huang2021power}, raising concerns about the viability of quantum kernel methods~\cite{pop00003}. Additionally, calculating numerical gradients on noisy qubits presents a delicate balancing act: reducing the step size to improve accuracy can obscure subtle differences in the cost function for nearby parameter values~\cite{pop00021}.
This complexity necessitates a deeper exploration into controlled noise simulations, such as the depolarization channel, to better understand and counteract these effects.

In the worst-case scenario, we can use the depolarization channel to simulate the quantum system's noise~\cite{pop00003}. However, executing depolarizing noise in a controlled manner on quantum hardware presents both critical challenges and intriguing opportunities for advancing our understanding and mitigation of this noise model. One of the primary challenges in executing depolarizing noise lies in its inherently probabilistic nature. Depolarization introduces errors with a certain probability, often modeled by the Kraus operators, onto the quantum state~\cite{nielsen2010quantum}. Implementing such probabilistic errors precisely on hardware requires sophisticated control techniques and careful calibration procedures. Inaccurate noise injection can lead to deviations from the expected noise model, compromising the validity of subsequent experiments and analyses. To navigate these challenges and unlock the full potential of QML, the development of robust error correction techniques is paramount~\cite{Wootton2012High}. Techniques like surface codes~\cite{fowler2012surface} and stabilizer codes~\cite{gottesman1997stabilizer} offer promising avenues for mitigating the detrimental effects of noise and safeguarding the integrity of quantum computations. Using noise-estimate circuits combined with other error correction techniques, we can estimate and correct such noise even in the circuits with extensive CNOT gates~\cite{urbanek2021mitigating}. We can further enhance the accuracy of error mitigation using multi-exponential error mitigation techniques~\cite{Cai2020Multi-exponential}. While using this technique, we may be able to model depolarizing noise accurately, but the computational cost can be prohibitive in the NISQ era~\cite{pop00004}.

The conventional model for representing depolarization noise employs three Pauli matrices, $(X, Y, \text{ and } Z)$, to capture isotropic noise processes affecting quantum states~\cite{nielsen2010quantum}. While comprehensive, this standard approach requires three Kraus operators corresponding to each Pauli matrix, resulting in six total matrix multiplications to simulate the noise channel. However, as discussed by~\cite{preskill2018quantum}, this mathematical formalism introduces significant computational overhead, particularly in near-term systems where hardware resources are limited. 

To address these challenges, we propose a modified approach for single-qubit depolarization utilizing only two Kraus operators and the $X$ and $Z$ Pauli matrices. This reduced model not only simplifies the mathematical representation but also decreases the computational complexity from six to four matrix multiplications for each noise channel execution. Such an approach provides a more efficient means of simulating depolarization in resource-constrained quantum hardware, an essential capability in the NISQ era where computational resources are scarce. By developing simplified yet representative noise models, our work aims to enable more efficient and scalable approaches to simulating and correcting depolarization noise in deep quantum circuits.

The rest of the paper is organized as follows. Section~\ref{sec:derivation} provides background on the standard depolarizing channel and derives the proposed modified channel representation.  Section~\ref{sec:exp} experimentally analyzes the modified channel on the QML task on the Iris dataset. Section~\ref{sec:discussion} discusses the implications of our method, and section~\ref{sec:conclusion} concludes with a summary of our contributions and an outlook on future research directions.

\section{Derivation}\label{sec:derivation}

\subsection{Standard Expression of the Depolarizing Channel}

The depolarizing channel represents a quantum noise where the qubit state is replaced by a completely mixed state with a certain probability~\cite{nielsen2010quantum}. The standard expression of the depolarizing channel $\mathcal{E}$ for a single qubit is given by:
\begin{equation} \label{eq:generalDepolarize}
 \rho' = \mathcal{E}(\rho) = (1 - p) \rho + \frac{p}{3} (X \rho X + Y \rho Y + Z \rho Z)
\end{equation}
where $\rho$ is the density matrix; $p$ is the depolarization rate; and $X$, $Y$, and $Z$ are the Pauli matrices, for the $X$, $Y$, and $Z$ quantum gates, respectively.

Eq.~(\ref{eq:generalDepolarize}) implies that with probability $(1 - p)$, the qubit state remains unchanged, and with probability $p$, it is subjected to equal mixtures of bit-flip, phase-flip, and both bit and phase-flip errors.

We can represent~(\ref{eq:generalDepolarize}) using Kraus operators as:
\begin{equation} \label{eq:generalkraus}
    \rho' = K_0 \rho K_0^\dag + K_1 \rho K_1^\dag + K_2 \rho K_2^\dag + K_3 \rho K_3 ^\dag = \sum_i K_i \rho K_i^\dag
\end{equation}
where, $$
\begin{array}{llll}
K_0 = \sqrt{1-p}\mathbb{I},&
K_1 = \sqrt{\frac{p}{3}}X, &
K_2 = \sqrt{\frac{p}{3}}Y, &
K_3 = \sqrt{\frac{p}{3}}Z 
\end{array}
$$
$\mathbb{I}$ is an identity matrix. Having defined the standard depolarizing channel and associated Kraus operators, we will next derive an alternative representation of this channel.

\subsection{Alternative Expression of the Depolarizing Channel}
We introduce an alternative representation of the depolarization channel characterized by reduced matrix multiplication operations that only use the $X$ and $Z$ Pauli matrices. We define this alternative representation of the depolarizing channel as:
\begin{equation} \label{eq:modifiedDepolarize}
     \rho_{m}' = (1 - \frac{2p}{3}) \rho + \frac{2p}{3} Z((\rho X)^T X) Z
\end{equation}
In this representation, the state is partly retained with a coefficient of $(1 - \frac{2p}{3})$ and partly subjected to a specific combination of Pauli $X$ and $Z$ operations with a coefficient of $\frac{2p}{3}$. This alternative expression is validated below to produce the same results as~(\ref{eq:generalDepolarize}).

\begin{theorem}
Eq. (\ref{eq:generalDepolarize}) and (\ref{eq:modifiedDepolarize}) are equivalent.
\end{theorem}

\begin{proof}
Consider the Pauli Matrices:
\[ \begin{array}{cccc}
    X = \begin{bmatrix}
    0 & 1 \\
    1 & 0
    \end{bmatrix},& \hspace{-0.5em}
    Y = \begin{bmatrix}
        0 & -i \\
        i & 0
    \end{bmatrix},&\hspace{-0.5em}
    Z = \begin{bmatrix}
        1 & 0 \\
        0 & -1
    \end{bmatrix},& \hspace{-0.5em}
    \mathbb{I} = \begin{bmatrix}
        1 & 0 \\
        0 & 1
    \end{bmatrix}
\end{array} \]

Let us consider an arbitrary density matrix $\rho = \begin{bmatrix} a & b \\ c & d \end{bmatrix}$ for a single qubit quantum state.
Substituting $\rho$ in (\ref{eq:generalDepolarize}) and ~(\ref{eq:modifiedDepolarize}) and with trivial algebraic work, we get:
    \begin{equation} \label{eq:simplifiedeq}
    \rho' = \rho_{m}' = 
        \begin{bmatrix}
            -2a\frac{p}{3} + 2ad\frac{p}{3} & -4b\frac{p}{3} + b \\
            -4c\frac{p}{3} + c & -2d \frac{p}{3} + d .
        \end{bmatrix}
    \end{equation}
Hence, it can be seen that (\ref{eq:generalDepolarize}) and (\ref{eq:modifiedDepolarize}) is the same for a single qubit and for an arbitrary $\rho$. 
\end{proof}
Next, we will define Kraus operators and prove their validity. 

\begin{theorem}\label{th:modifiedKraus}
The following Kraus operators for (\ref{eq:modifiedDepolarize}) are valid operators.
$$
\begin{array}{cc}
    K_0 = \sqrt{1 -\frac{2p}{3}} \mathbb{I}, &
    K_1 = i \sqrt{\frac{2p}{3}} ZX .
\end{array}
$$
\end{theorem}

\begin{proof}
From (\ref{eq:modifiedDepolarize}), one can immediately see that the corresponding Kraus operator corresponding to the term $(1 - \frac{2p}{3}) \rho$ is: $\sqrt{1 -\frac{2p}{3}}{\mathbb{I}}$.
Now let us consider the second terminology of (\ref{eq:modifiedDepolarize}), i.e., $\frac{2p}{3}Z((\rho X)^T X ) Z$, which enables us to re-write without loss of generality the following:
\begin{multline*}
\frac{2p}{3} Z((\rho X )^T X) Z = \frac{2p}{3} ZX\rho XZ \quad\\
(\because \rho = \rho^T \text{ and } AX^T = XA),\text{ for any }A \in \mathbb{R}^{2 \times 2} 
\end{multline*}

Next, we want a Kraus operator \( K_1 \) s.t. \( K_1 \rho K_1^\dag = ZX \rho XZ \)

Thus, intuitively,
$$ 
\begin{array}{cc}
K_1 \propto ZX, & 
K_1 = x ZX, \quad (x \text{ is a scalar}).
\end{array}
$$
Following the Kraus operator completeness constraint, we can write:
\begin{equation*}
K_0 K_0^{\dagger} + K_1 K_1^{\dagger} = \mathbb{I} ,
\end{equation*}
or
\begin{equation} \label{eq:modifiedK1validity}
K_1 K_1^{\dagger} = \frac{2p}{3} \mathbb{I}
\end{equation}

To satisfy (\ref{eq:modifiedK1validity}), \( K_1 \) must have a magnitude of $\sqrt{\frac{2p}{3}}$. Therefore, 
\begin{equation*}
K_1 = \sqrt{\frac{2p}{3}} ZX
\end{equation*}

We added \say{$i$} to correct a sign discrepancy while validating the operators, resulting in:
\begin{equation*} \label{eq:KrausModifiedK1}
K_1 = i \sqrt{\frac{2p}{3}} ZX
\end{equation*}

\subsubsection*{Validation of Kraus operators}

Any set of Kraus operators satisfies the completeness property. That is,
$$\sum_i K_i^\dag K_i = K_0^\dag K_0 + K_1^\dag K_1 =  \mathbb{I}$$
Solving each of the Kraus operators squared individually, we can get,
$$
\begin{array}{cc}
    K_0^\dag K_0 = (1 - \frac{2p}{3})\mathbb{I}, &
    K_1^\dag K_1 = \frac{2p}{3}\mathbb{I}
\end{array}$$
$$
    \therefore K_0^\dag K_0 + K_1^\dag K_1 = (1 - \frac{2p}{3})\mathbb{I} + \frac{2p}{3}\mathbb{I} = \mathbb{I}
$$
This proves that the Kraus operators proposed in theorem~(\ref{th:modifiedKraus}) are valid Kraus operators. 
\end{proof}

It follows from~(\ref{eq:generalkraus}) that we can re-define (\ref{eq:modifiedDepolarize}) as:
\begin{equation}\label{eq:modifiedKraus}
    \rho_{m}' = K_0 \rho K_0^\dag + K_1 \rho K_1^\dag
\end{equation}

In general, for a single qubit representation (\ref{eq:generalDepolarize}), (\ref{eq:generalkraus}), (\ref{eq:modifiedDepolarize}), and (\ref{eq:modifiedKraus}) yield the same result.

The above derivations show that the modified depolarization channel expression is equivalent to the standard equation. We further proved that the proposed Kraus operators for (\ref{eq:modifiedDepolarize}) are valid Kraus operators. The next step is to show that the matrix given by the modified channel is a valid density matrix. To do this, we need to prove that (\ref{eq:modifiedDepolarize}) is Hermitian, positive semi-definite, and has a unit trace.

\begin{theorem}
The matrix given by (\ref{eq:modifiedDepolarize}) is a valid density matrix.
\end{theorem}

\begin{proof}
First, let us shows that \( \rho_{m}' \) is Hermitian.

A matrix is Hermitian if it equals its own conjugate transpose, \( A = A^\dagger \).

To show \( \rho_{m}' \) is Hermitian, we calculate its conjugate transpose using its definition given by~(\ref{eq:modifiedKraus}):
\[
\begin{aligned}
(\rho_{m}')^\dagger &= (K_0 \rho K_0^\dag + K_1 \rho K_1^\dag)^\dagger \\
&= (K_0 \rho K_0^\dag)^\dagger + (K_1 \rho K_1^\dag)^\dagger \\
&= (K_0^\dag)^\dagger \rho^\dagger K_0^\dagger + (K_1^\dag)^\dagger \rho^\dagger K_1^\dagger \quad (\because (AB)^\dagger = B^\dagger A^\dagger ) \\
&= K_0 \rho K_0^\dag + K_1 \rho K_1^\dag \\
&= \rho_{m}' .
\end{aligned}
\]

Since \( K_0 \) and \( K_1 \) are constructed from unitary matrices and complex numbers, their conjugate transposes are simply their adjoints, hence \( \rho_{m}' \) is Hermitian.

Second, let us show that \( \rho_{m}' \) is Positive Semi-Definite.

Given the Kraus operators \(K_0\) and \(K_1\), and density matrix \( \rho_{m}' \)
we want to show that for any vector \(v\), the expectation value \(v^\dagger \rho_{m}' v\) is non-negative.

Starting with the expression for \(v^\dagger \rho_{m}' v\):
\[
v^\dagger \rho_{m}' v = v^\dagger (K_0 \rho K_0^\dagger + K_1 \rho K_1^\dagger) v.
\]
This expands to:
\[
v^\dagger \rho_{m}' v = v^\dagger K_0 \rho K_0^\dagger v + v^\dagger K_1 \rho K_1^\dagger v.
\]
We can express each term as:
\[
v^\dagger K_0 \rho K_0^\dagger v = (K_0^\dagger v)^\dagger \rho (K_0^\dagger v),
\]
\[
v^\dagger K_1 \rho K_1^\dagger v = (K_1^\dagger v)^\dagger \rho (K_1^\dagger v)
\]
where \(w_0 = K_0^\dagger v\) and \(w_1 = K_1^\dagger v\) are vectors in the Hilbert space.

Since \(\rho\) is a positive semi-definite density matrix, we have for any vector \(w\), \(w^\dagger \rho w \geq 0\). Applying this to \(w_0\) and \(w_1\):
\[
(w_0)^\dagger \rho w_0 \geq 0,
\]
\[
(w_1)^\dagger \rho w_1 \geq 0.
\]
Therefore, the sum is also non-negative:
\[
v^\dagger \rho_{m}' v = (w_0)^\dagger \rho w_0 + (w_1)^\dagger \rho w_1 \geq 0,
\]
which establishes that \( \rho_{m}' \) is positive semi-definite.

Third, let us show that \( \rho_{m}' \) has Unit Trace ,i.e., \( \text{Tr}(\rho_{m}') = 1 \).


The trace of \( \rho_{m}' \) is given by:
\[
\text{Tr}(\rho_{m}') = \text{Tr}(K_0 \rho K_0^\dag) + \text{Tr}(K_1 \rho K_1^\dag).
\]
Using the cyclic property of the trace, we can rewrite this as:
\[
\text{Tr}(\rho_{m}') = \text{Tr}(K_0^\dag K_0 \rho) + \text{Tr}(K_1^\dag K_1 \rho).
\]
Computing \( K_0^\dag K_0 \) and \( K_1^\dag K_1 \), we get:
\[
K_0^\dag K_0 = (1 -\frac{2p}{3}) \mathbb{I},
\]
\[
K_1^\dag K_1 = -\frac{2p}{3} \mathbb{I},
\]

Thus, the trace of \( \rho_{m}' \) simplifies to:
\[
\text{Tr}(\rho_{m}') = \text{Tr}((1 -\frac{2p}{3}) \mathbb{I} \rho) - \text{Tr}(\frac{2p}{3} \mathbb{I} \rho).
\]
As \( \rho \) is a density matrix with unit trace, \( \text{Tr}(\rho) = 1 \), this leads to:
\[
\text{Tr}(\rho_{m}') = (1 -\frac{2p}{3}) - \frac{2p}{3} = 1 - 2p + 2p = 1.
\]
Therefore, \( \rho_{m}' \) maintains the unit trace property, as required for any density matrix.
\end{proof}

Next, we derive the expression on how (\ref{eq:modifiedDepolarize}) evolves when the depolarization channel is applied $m$ times on a quantum state $\rho$ that leads us to the following lemma. 

\begin{lemma}
When a depolarization channel with $p$ depolarizing rate is applied on a quantum state $\rho$ up to $m$ times, the resulting quantum state is defined as follows up to first order in $p$:
\begin{equation}\label{eq:modifiedDepolarizeUptoMTimes}
\rho'_{mm} = (1-\frac{2mp}{3})\rho + \frac{2mp}{3} Z((\rho X)^T X) Z + \mathcal{O}(p^2)
\end{equation}
And, for an observable $O$, the expectation value is given as:
\begin{multline}\label{eq:expvalofmodifiedm}
    \langle O \rangle_{\rho'_{mm}} = \text{Tr}\{O \rho'_{mm} \} = \text{Tr}\{O \rho\} - \frac{2mp}{3} \text{Tr}\{O \rho\} \\
    + \frac{2mp}{3} \text{Tr}\{OZX^T\rho^TXZ\} .
\end{multline}
\end{lemma}
In the following section, we verify that the results obtained from (\ref{eq:expvalofmodifiedm}) and the standard channel simulations are negligible. We will demonstrate through experimental evidence that~(\ref{eq:modifiedDepolarize})and~(\ref{eq:modifiedDepolarizeUptoMTimes} and \ref{eq:expvalofmodifiedm}) can be effectively used for training machine learning models.

\section{Experiment}\label{sec:exp}
We start this section by showing that the results from ~(\ref{eq:expvalofmodifiedm}) are consistent with simulation results of~(\ref{eq:generalDepolarize}) for multiple values of $m$ and $p$.  Later, we empirically show that the Depolarization rate up to threshold $w$ does not affect the performance of the single qubit QML model for the iris dataset. For the scope of this experiment, we are considering the binary classification. We used the first, Setosa,  and third, Virginia, flower classes and only the first 2 features, Sepal length and sepal width. We conducted the experiment on Python with Pennylane for quantum circuits simulations and quantum computation. We would like to mention that otherwise mentioned, we used ~(\ref{eq:modifiedDepolarize}) for the depolarization channel and ~(\ref{eq:expvalofmodifiedm}) for the depolarization channel applied $m$ times.

\subsection{Quantum Circuit behavior analysis under Depolarization channel up to m times}
\begin{figure*}
     \centering
     \begin{subfigure}[b]{0.32\textwidth}
         \centering
         \includegraphics[width=\textwidth]{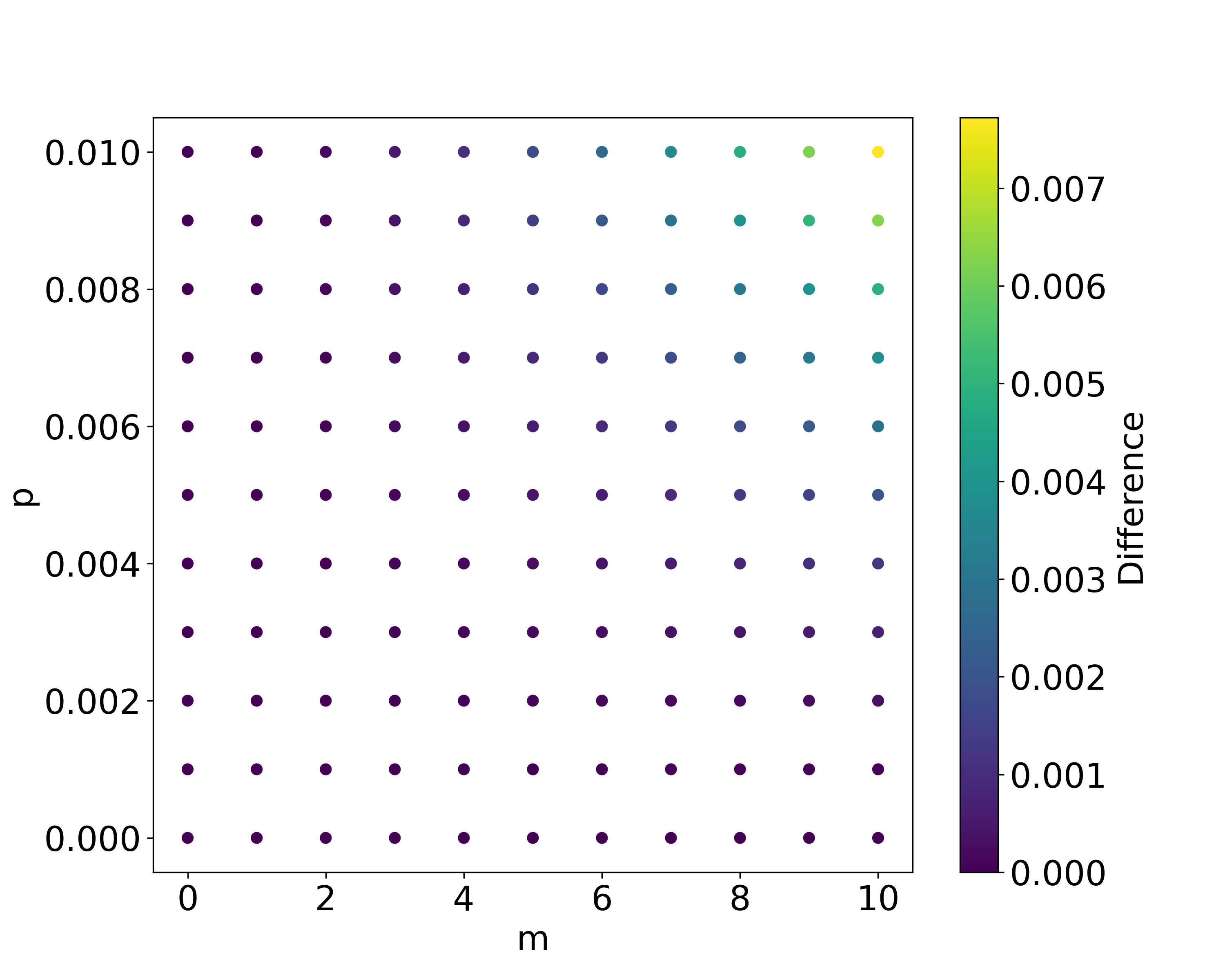}
         \caption{}
         \label{fig:multiple_p_m_3gates}
     \end{subfigure}
     \hfill
     \begin{subfigure}[b]{0.32\textwidth}
         \centering
         \includegraphics[width=\textwidth]{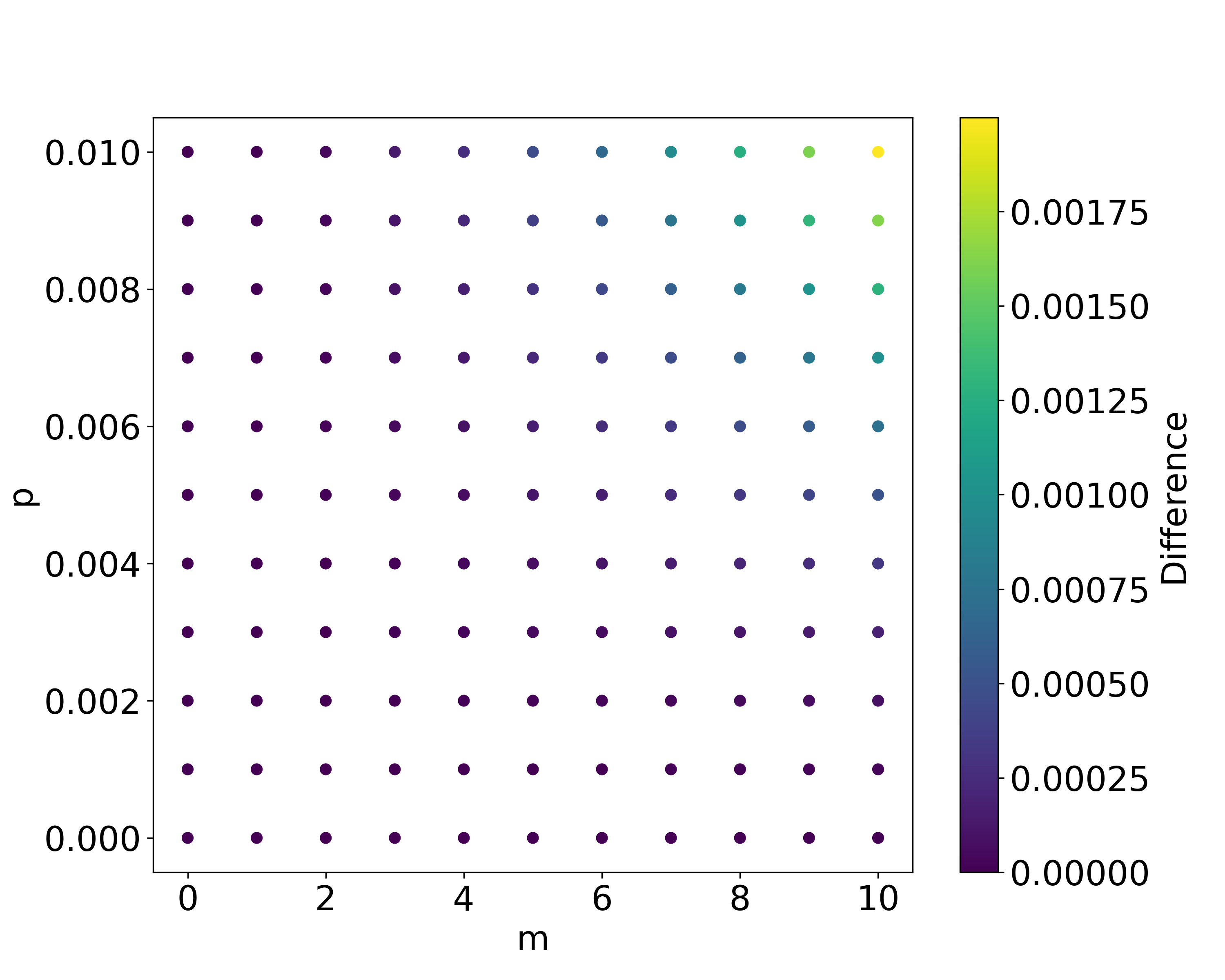}
         \caption{}
         \label{fig:multiple_p_m_8gates}
     \end{subfigure}
     \hfill
     \begin{subfigure}[b]{0.32\textwidth}
         \centering
         \includegraphics[width=\textwidth]{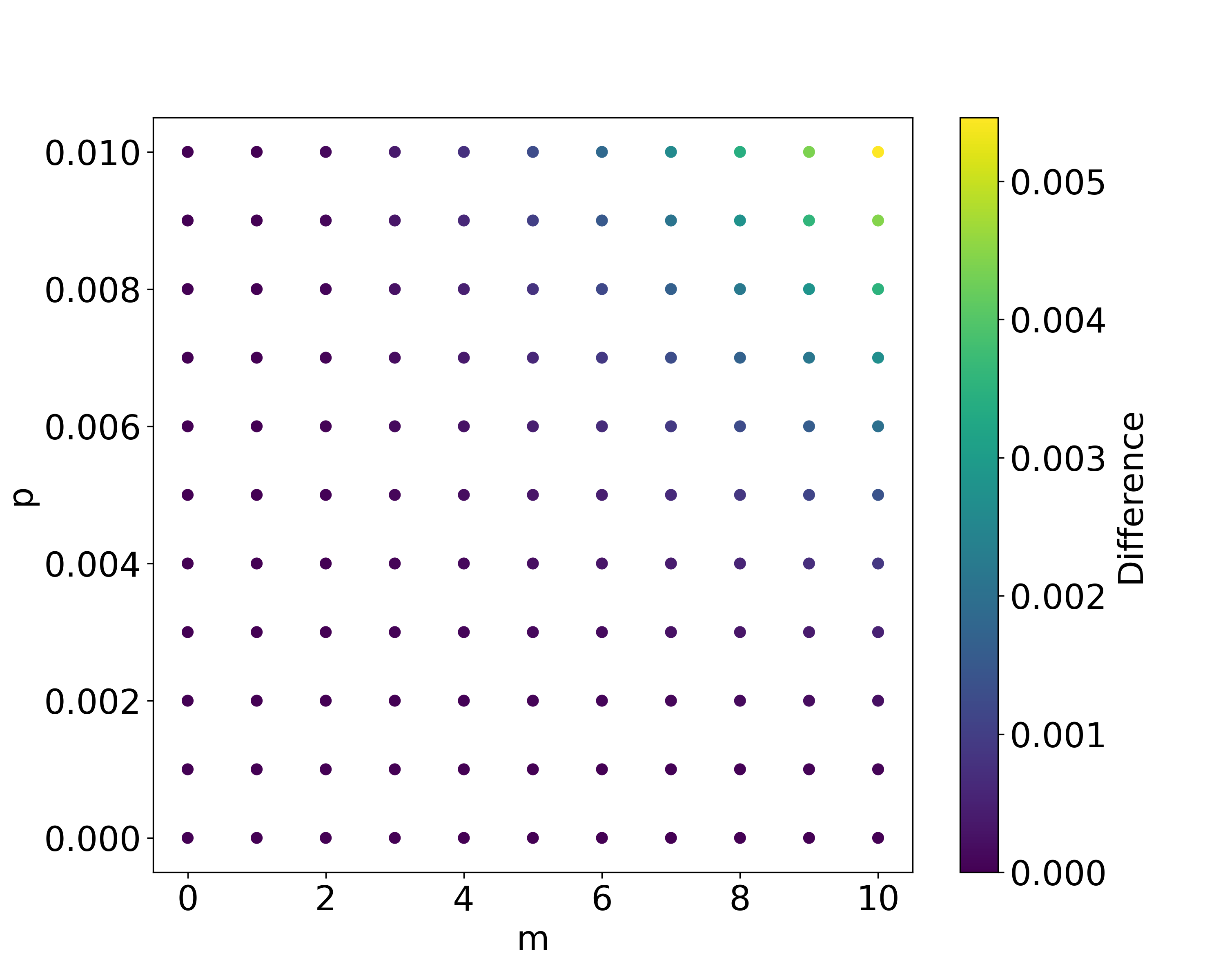}
         \caption{}
         \label{fig:multiple_p_m_15gates}
     \end{subfigure}
        \caption{Scatter plots present the difference between the standard channel and modified depolarization channel expectation value. Each channel was applied to a quantum circuit with single qubit gates of $3$, $8$, and $15$ respectively. The result for $3$ single qubit gates is presented in plot (a), while plot (b) and plot (c) represent the results for $8$ and $15$ gates circuits, respectively. The x-axis of each plot represents the number of times the noisy channel was applied and is given by $m$, while the y-axis gives the varying depolarization rates.}
    \label{fig:p_m_analysisofmultiplegatessize}
\end{figure*}
Eq.~(\ref{eq:expvalofmodifiedm}), posits that when a depolarization channel with a rate $p$ is applied $m$ times to a quantum state $\rho$, the resultant state $\rho'_{mm}$ adheres to a predictable transformation that maintains linearity with respect to $p$ in the first order for small value of $p$. This suggests that despite the iterative application of noise, the overall system's behavior under the depolarization channel can be approximated linearly. We conduct a series of simulations to substantiate this theory. For this, we computed the expectation values for a single qubit with varying depths—specifically, $3, 8,$ and $15$ gates. The behavior is assessed across different depolarization rates, $p = [0.0, 0.001,0.002,0.003,0.004,0.005,0.006,0.007,0.008, \\ 
0.009,0.01]$, and depolarization channel repetition, $m = [0, 1, 2, 3, 4, 5, 6, 7, 8, 9, 10]$.

Fig.~\ref{fig:p_m_analysisofmultiplegatessize} presents a scatter plot visualization for the expectation value differences between~(\ref{eq:modifiedDepolarizeUptoMTimes}) and the standard depolarization model as a function of both $p$ and $m$.  There is a minimal deviation between the standard and modified channels' expectation value for low depolarization rates across all gate counts. This alignment implies that the modified equation retains fidelity to the standard model's predictions in the low-noise regime. The plots exhibit a uniform trend where, for small values of $p$, the difference in expectation values is negligible across all values of $m$. This negligible difference remains consistent as the number of gates increases, emphasizing the robustness of the modified model. 
Extending on these results, we analyze its performance in QML by training the QML model on the Iris dataset under the modified depolarization channel. We map the classic data into quantum Hilbert space via a feature map.

\subsection{Data Encoding}
\begin{figure*}
  \centering
    \begin{subfigure}[b]{\textwidth} 
    \centering
    \includegraphics[width = \textwidth,trim={10cm 0 0 0},clip]{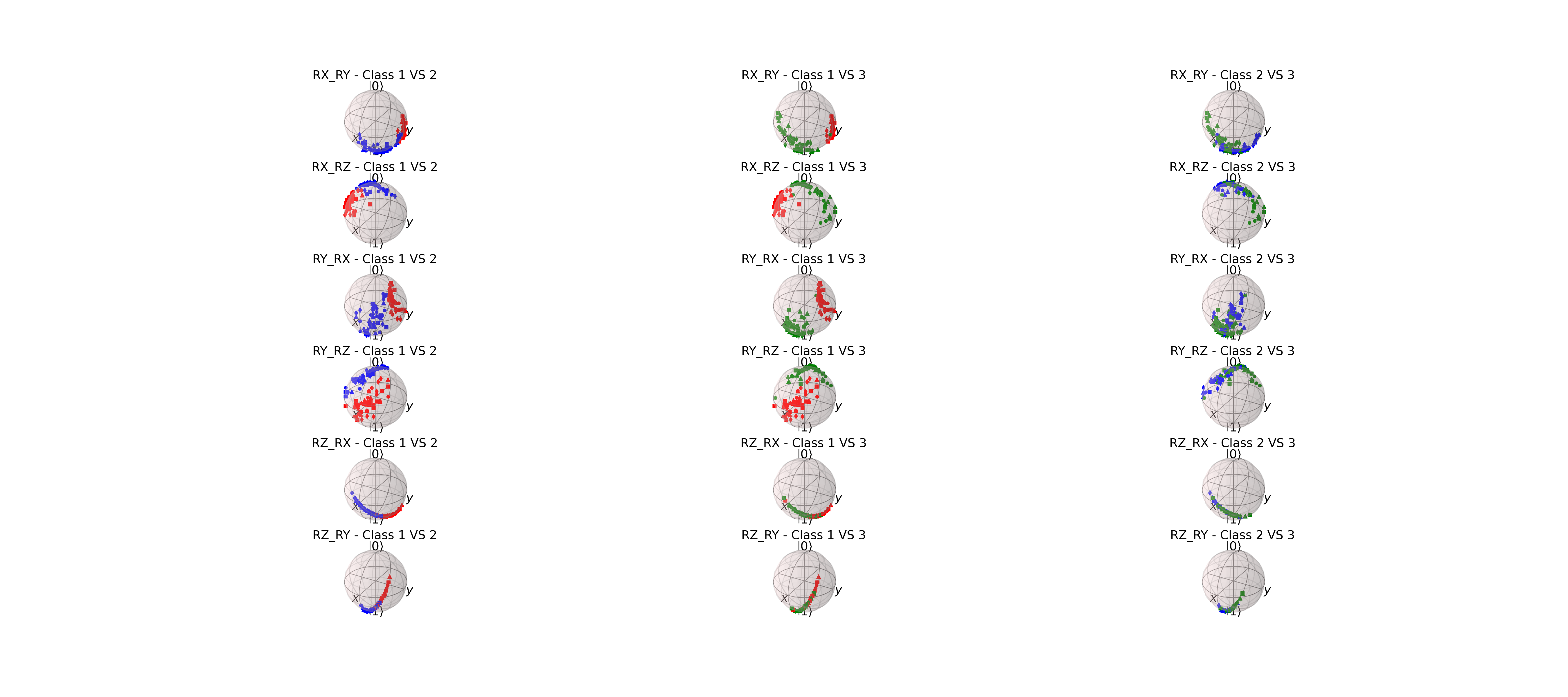}
    \caption{Various encoding schemes for single qubits using the rotational encoding. The combination of $RZ$ and $RX$ gates provides the best mapping for binary classification.}
    \label{fig:allencoding}
  \end{subfigure}
  \vfill
    \begin{subfigure}[b]{0.45\textwidth} 
    \centering
    \includegraphics[width=\textwidth,trim={0 9cm 0 9cm},clip]{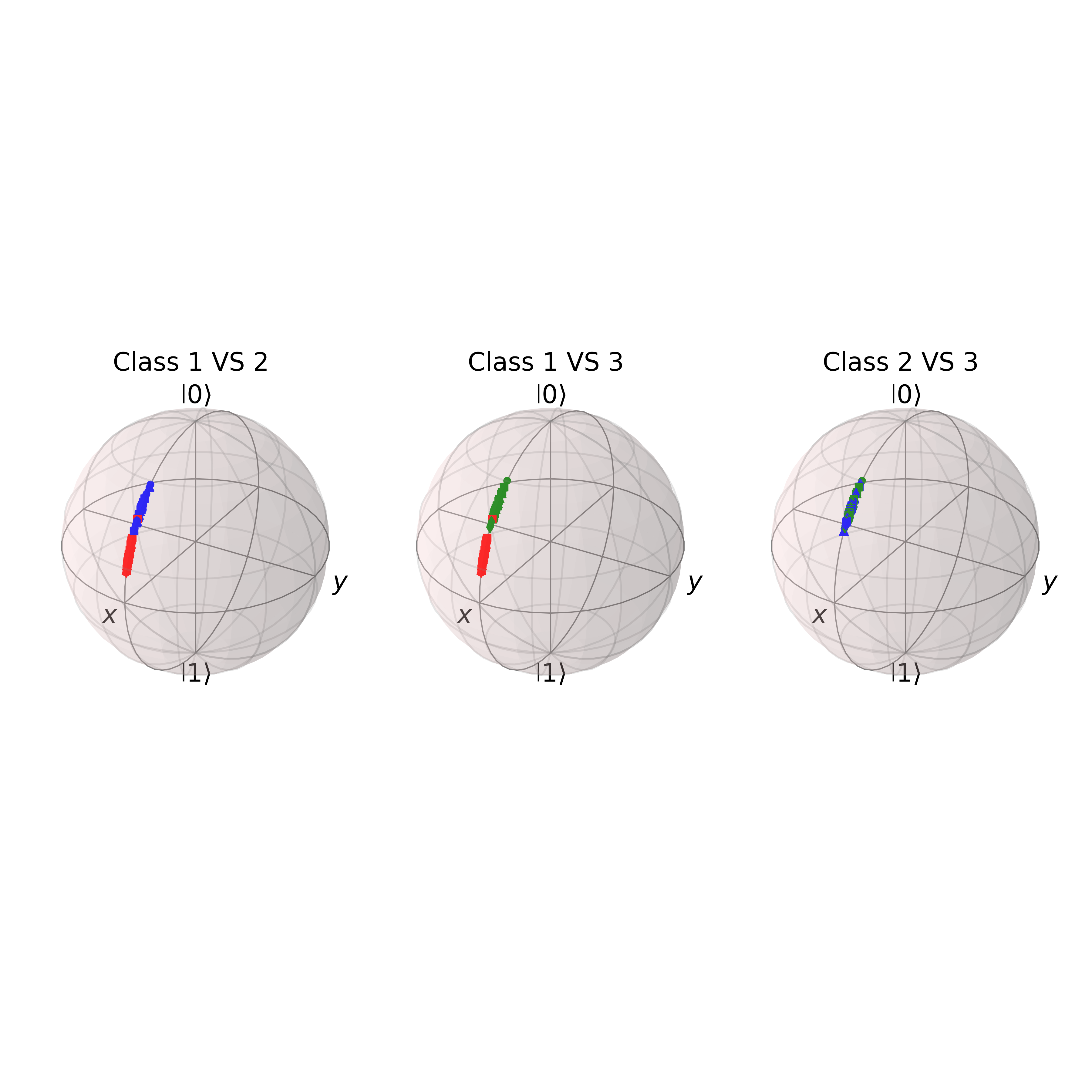} 
    \caption{Bloch Sphere representation of the quantum states obtained by Amplitude encoding of the features vectors}
    \label{fig:AmplitudeEncoding}
  \end{subfigure}
  \hfill 
  \begin{subfigure}[b]{0.45\textwidth} 
    \centering
    \includegraphics[width=\textwidth,trim={0 9cm 0 9cm},clip]{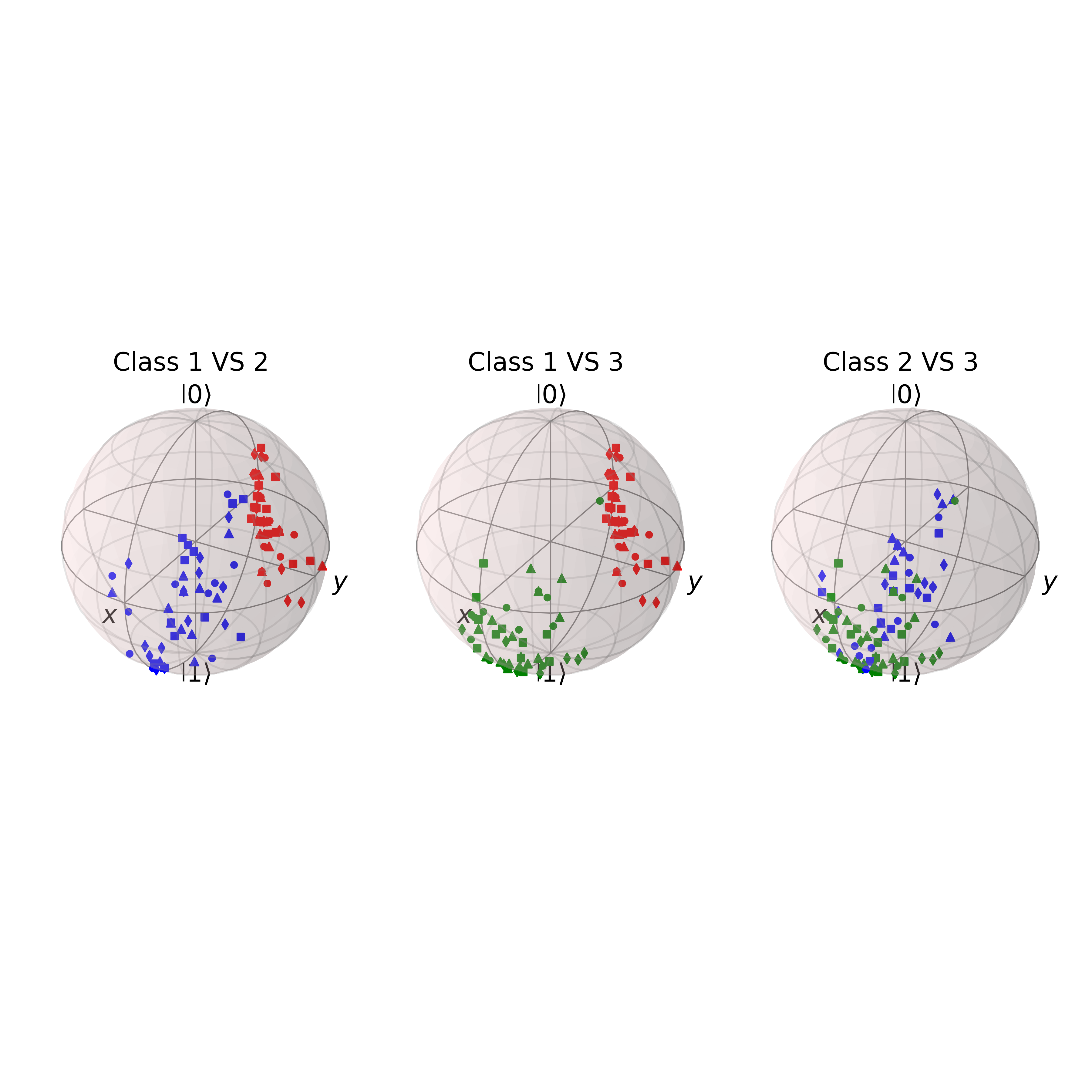} 
    \caption{Bloch Sphere representation of the quantum states obtained by Angle encoding of the features vectors}
    \label{fig:RotationEncoding}
  \end{subfigure}
  
  \caption{Feature Mapping of the Iris dataset using Amplitude Encoding and Rotational encoding method. The Rotational encoding scheme, a combination of $RX$ and $RY$, provides better mapping results for the classification problem.}
  \label{fig:AmpandRotEncoding}
\end{figure*}

We start by initializing the qubit in a computational basis $\ket{0} = \begin{bmatrix}
    1\\
    0
\end{bmatrix}$. Let $\phi:x\mapsto \phi(x)$ be a mapping from input space $X$ to a quantum Hilbert Space $\mathbb{R}$. Let us define $\phi(x) = RX(x_1)RY(x_0)$, where $x_0$ and $x_1$ are features of the input vector, and $$
\begin{array}{cc}
RY(x_0) = \begin{bmatrix}
    cos(\frac{x_0}{2}) & -sin(\frac{x_0}{2}) \\
    sin(\frac{x_0}{2}) & cos(\frac{x_0}{2})
\end{bmatrix} , \\
RX(x_1) = \begin{bmatrix}
    cos(\frac{x_1}{2}) & -i sin(\frac{x_1}{2}) \\
    -i sin(\frac{x_1}{2}) & cos(\frac{x_1}{2}) ,
\end{bmatrix}
\end{array}$$are rotational single qubit quantum gates. 
We chose the angle encoding scheme because it linearly separates input data in the Bloch Sphere better than amplitude encoding, as shown in Fig.~\ref{fig:AmpandRotEncoding}. Thus, for an input vector $x$, the data-encoded state is defined as:
\begin{equation}\label{eq:angleEncoding}
    \ket{\psi} = \phi(x) \ket{0} = RX(x_1)RY(x_0) \ket{0} .
\end{equation}

\subsection{Variational Layers}
Similar to the encoding scheme we applied a series of variational gates (RY and RX with parameters) whose parameters $\theta$ can be optimized during training. For $N$ trainable parameters $\theta$, we define the operation of Variational layers $U$ as:
\begin{equation} \label{eq:variationallayer}
    U(\theta) = \prod_{i=1}^N \text{Gate}_i(\theta_i) ,
\end{equation}
where $\text{Gate}_i(\theta_i)$ is either an $RY$ or $RX$ gate with parameter $\theta_i$. 
Thus, we can define a variational circuit as:
\begin{equation}\label{eq:VariationalCircuit}
    \ket{\Psi} = U(x,\theta) = U(\theta)\ket{\psi} .
\end{equation}
Let $\rho = \ket{\Psi}\bra{\Psi}$ be a density matrix. The system undergoes evolution through a depolarization channel. This channel, denoted as $\mathcal{E}(\rho)$, transforms the state be $\rho$ of the qubit by mixing it towards a maximally mixed state as the depolarization rate $p$ increases as given by equations~(\ref{eq:generalDepolarize}) and (\ref {eq:modifiedDepolarize}). Let this state $\rho'$. Now we measure an observable Pauli$-Z = \begin{bmatrix}
    1 & 0 \\
    0 & -1
\end{bmatrix}$ matrix to get a quantum machine learning model that can be defined as:

\begin{equation}\label{eq:qmlModel}
    f(x,\theta) = \text{Tr}\{Z \rho'\} .
\end{equation}
In the following section, we train this QML model on the Iris dataset under a modified depolarization channel and present its result.
\subsection{Training}
 \begin{figure*}
\begin{subfigure}[b]{0.52\textwidth}
\includegraphics[width=0.9\linewidth]{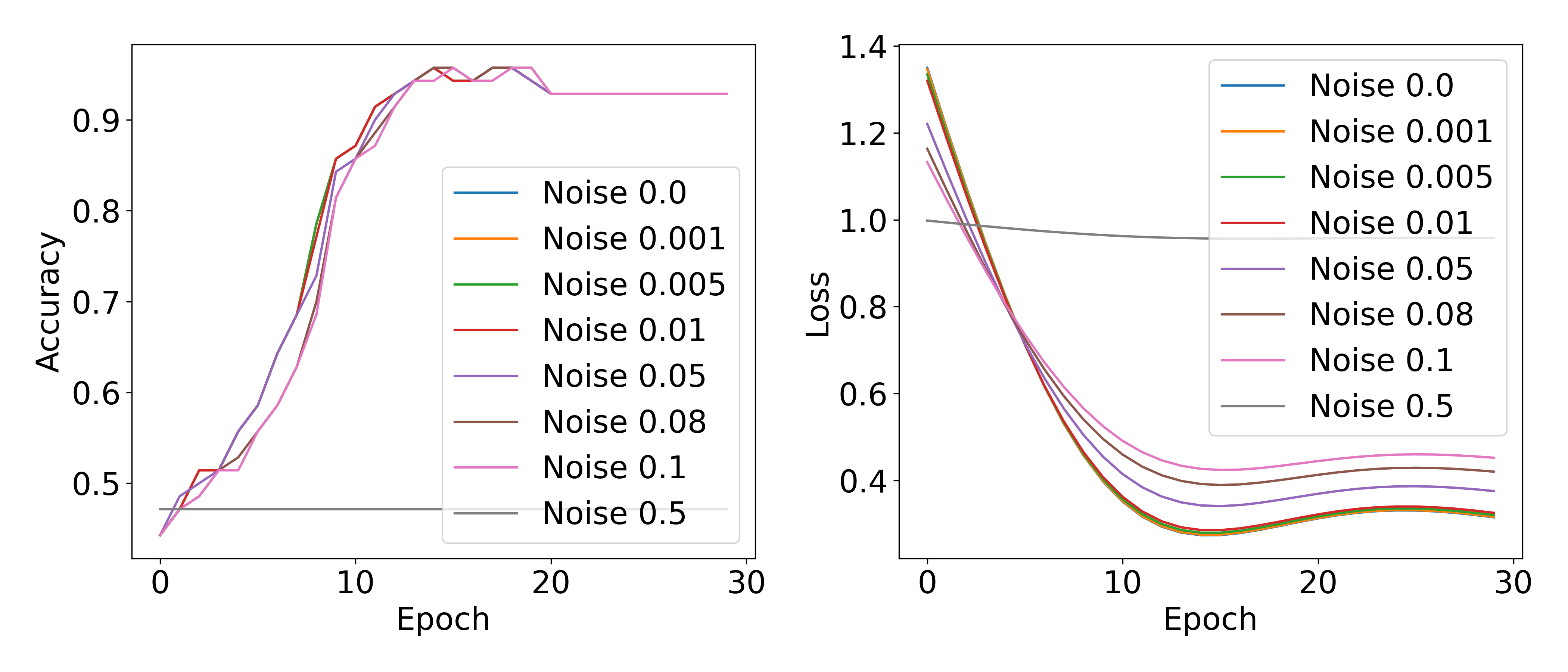} 
\caption{}
\label{fig:acc1}
\end{subfigure}
\hfill
\begin{subfigure}[b]{0.48\textwidth}
\includegraphics[width=0.9\linewidth]{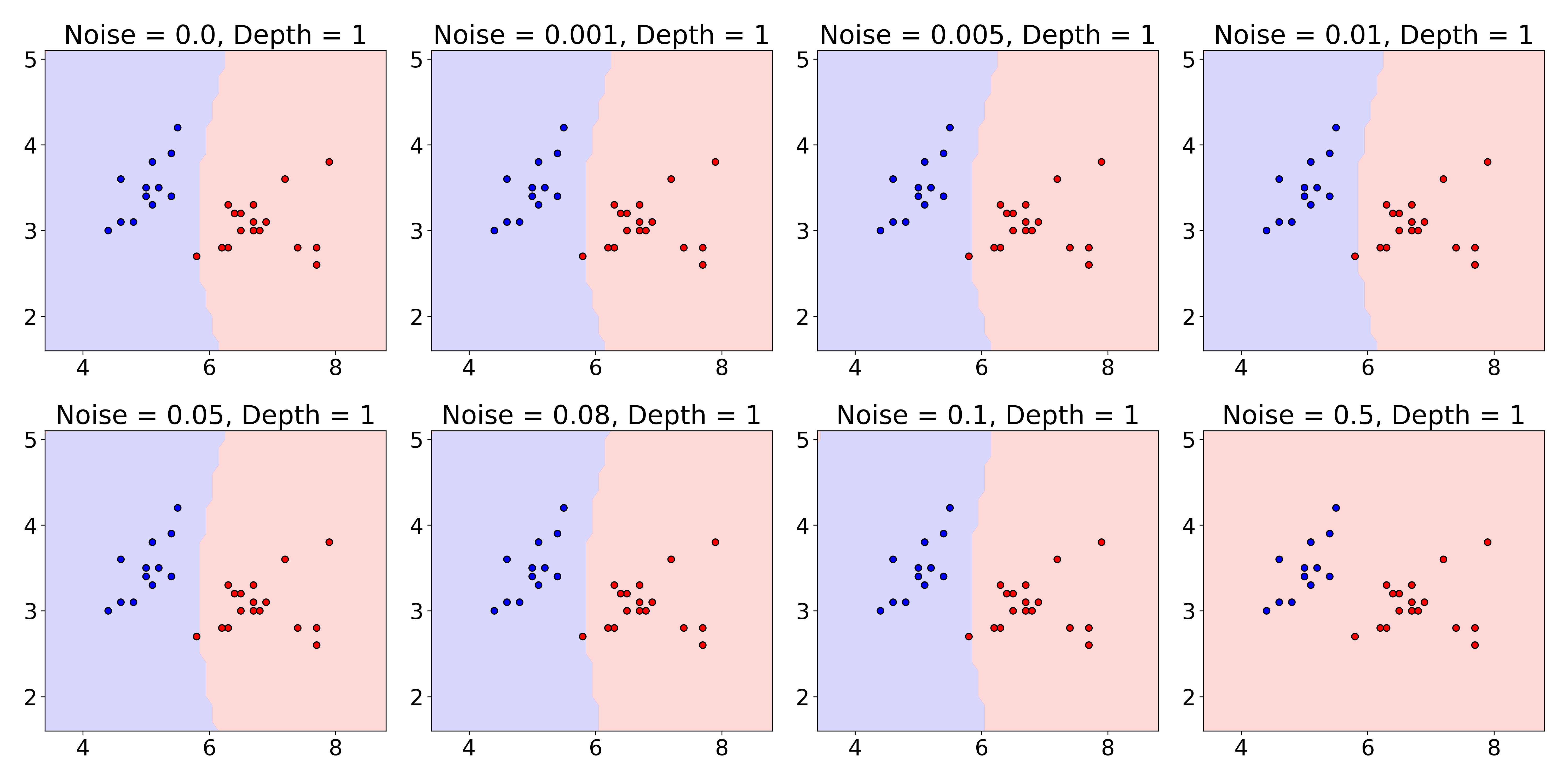}
\caption{}
\label{fig:dec1}
\end{subfigure}
\vfill

\begin{subfigure}[b]{0.52\textwidth}
\includegraphics[width=0.9\linewidth]{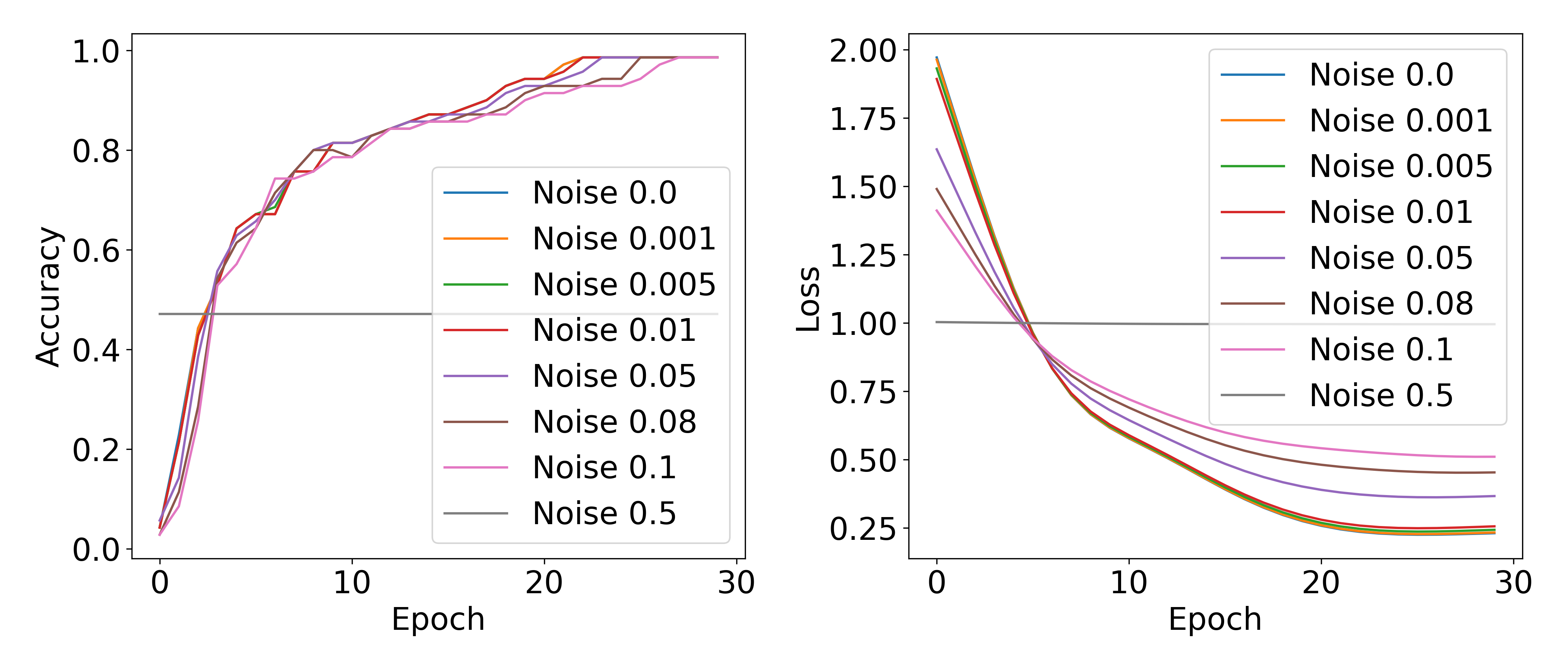} 
\caption{}
\label{fig:acc3}
\end{subfigure}
\hfill
\begin{subfigure}[b]{0.48\textwidth}
\includegraphics[width=0.9\linewidth]{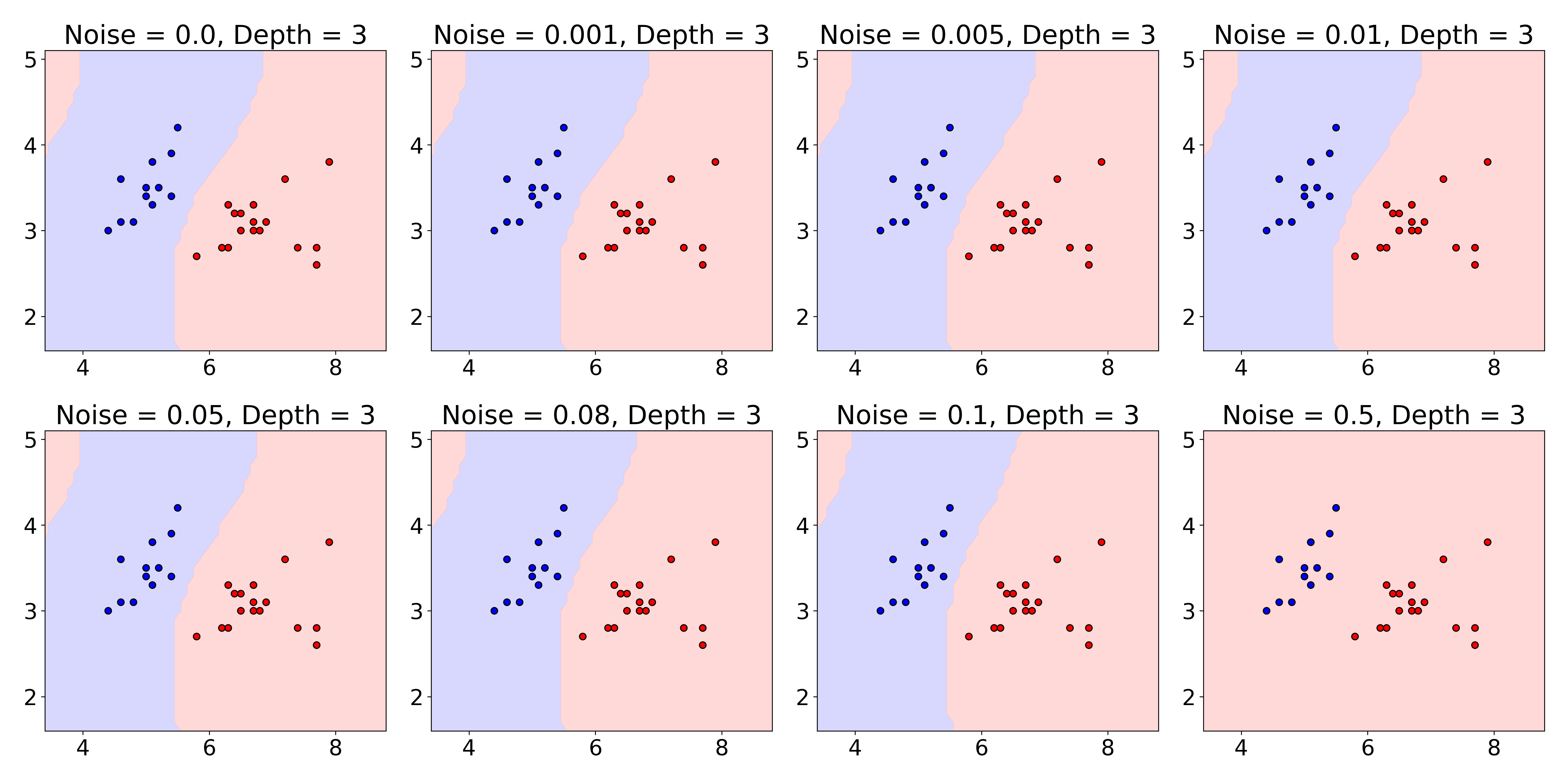}
\caption{}
\label{fig:dec3}
\end{subfigure}
\vfill

\begin{subfigure}[b]{0.52\textwidth}
\includegraphics[width=0.9\linewidth]{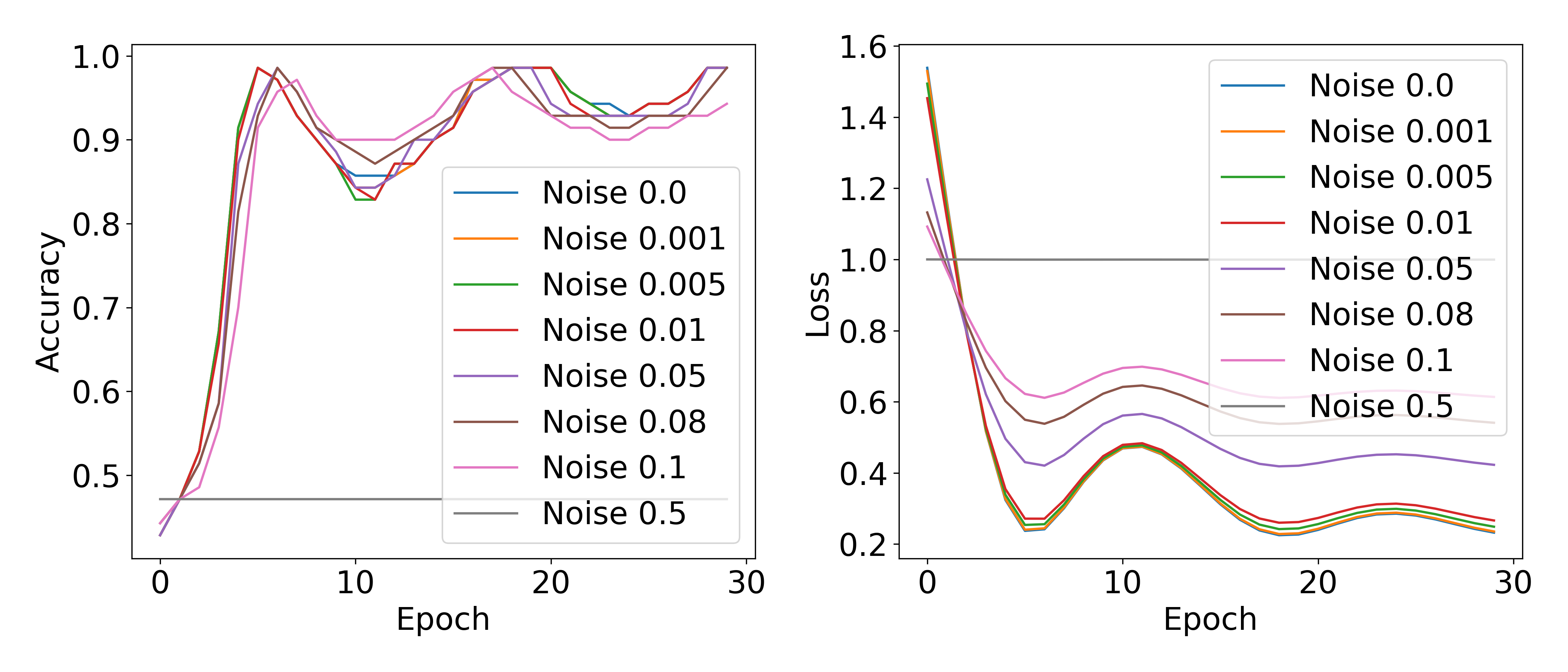} 
\caption{}
\label{fig:acc5}
\end{subfigure}
\hfill
\begin{subfigure}[b]{0.48\textwidth}
\includegraphics[width=0.9\linewidth]{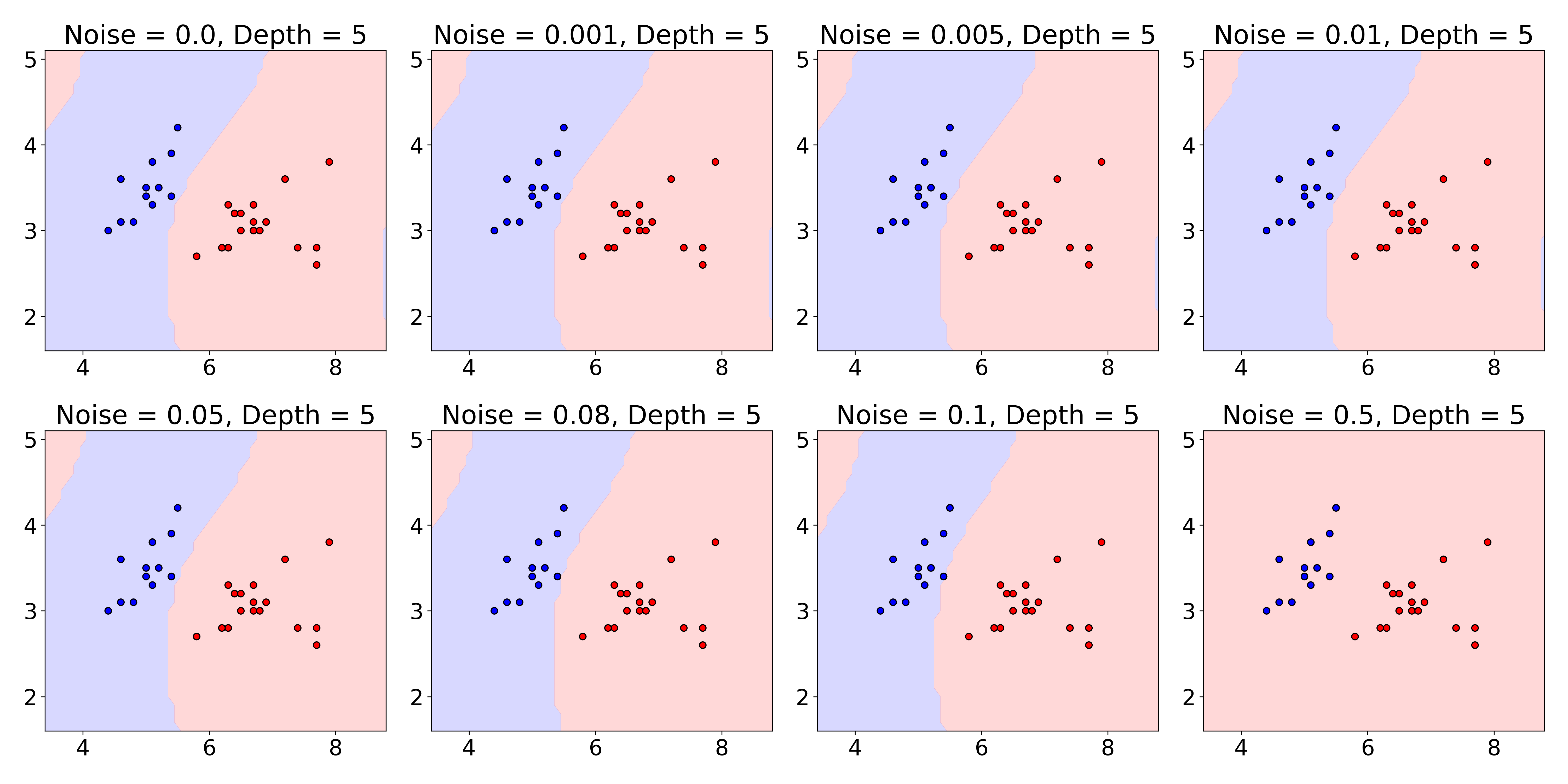}
\caption{}
\label{fig:dec5}
\end{subfigure}
\vfill

\begin{subfigure}[b]{0.52\textwidth}
\includegraphics[width=0.9\linewidth]{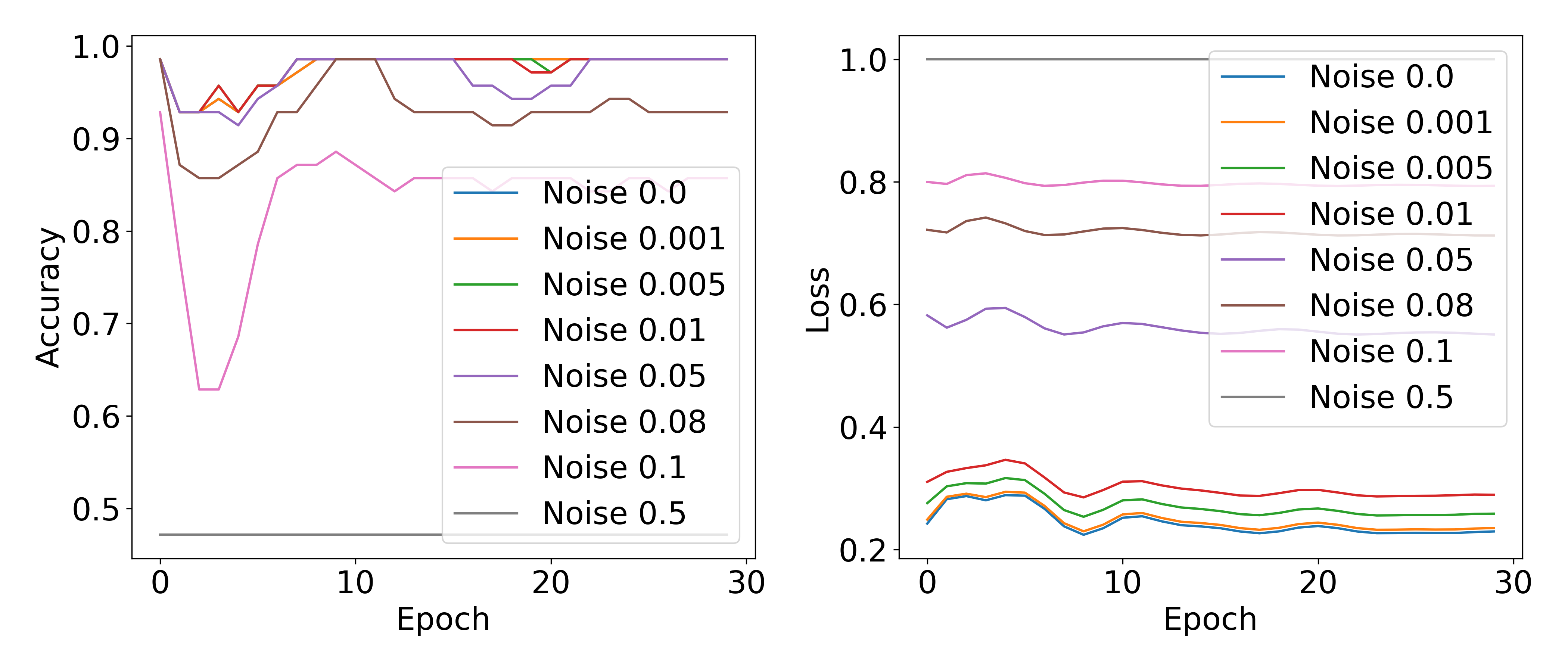} 
\caption{}
\label{fig:acc10}
\end{subfigure}
\hfill
\begin{subfigure}[b]{0.48\textwidth}
\includegraphics[width=0.9\linewidth]{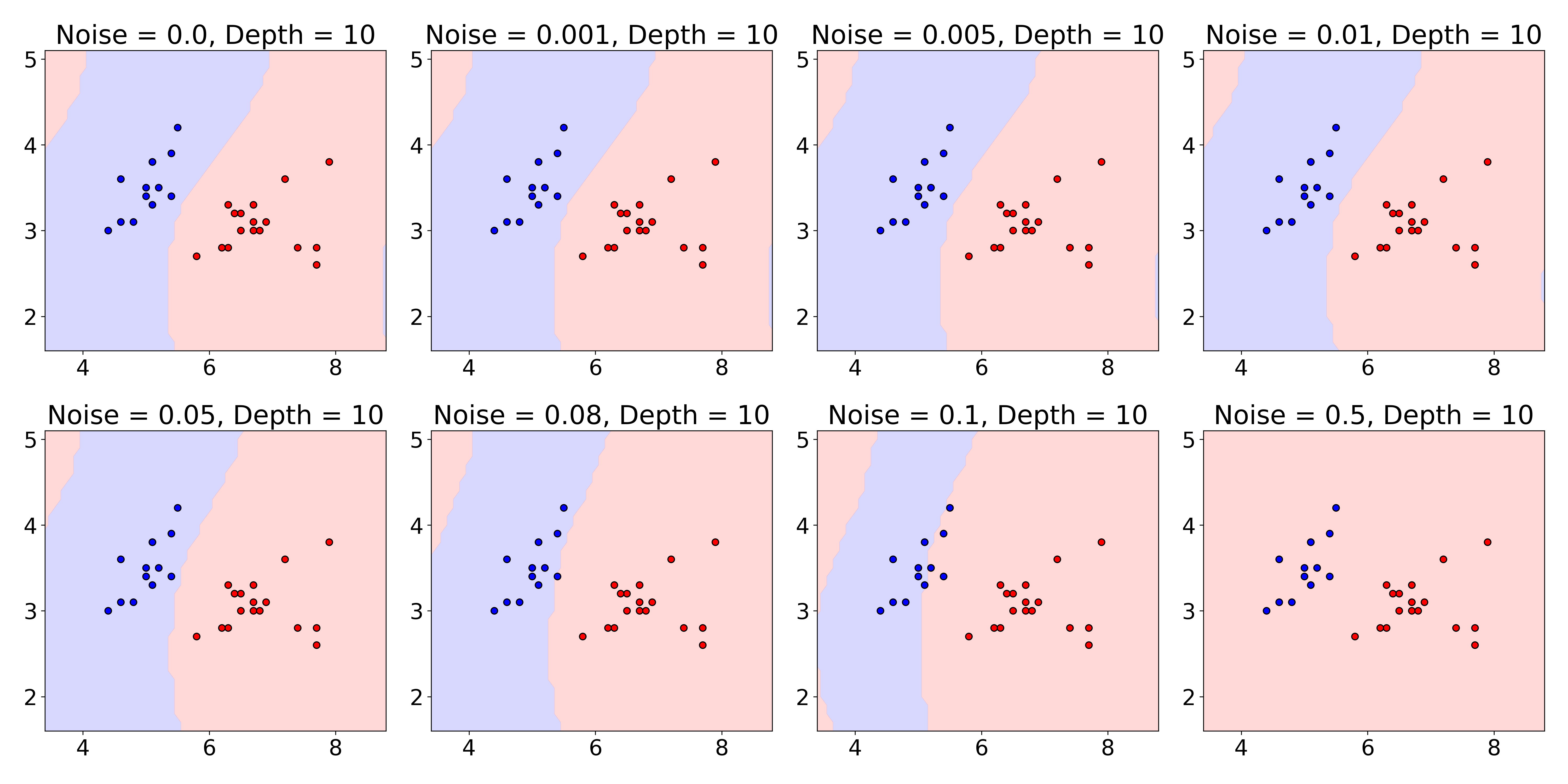}
\caption{}
\label{fig:dec10}
\end{subfigure}
\vfill

\begin{subfigure}[b]{0.52\textwidth}
\includegraphics[width=0.9\linewidth]{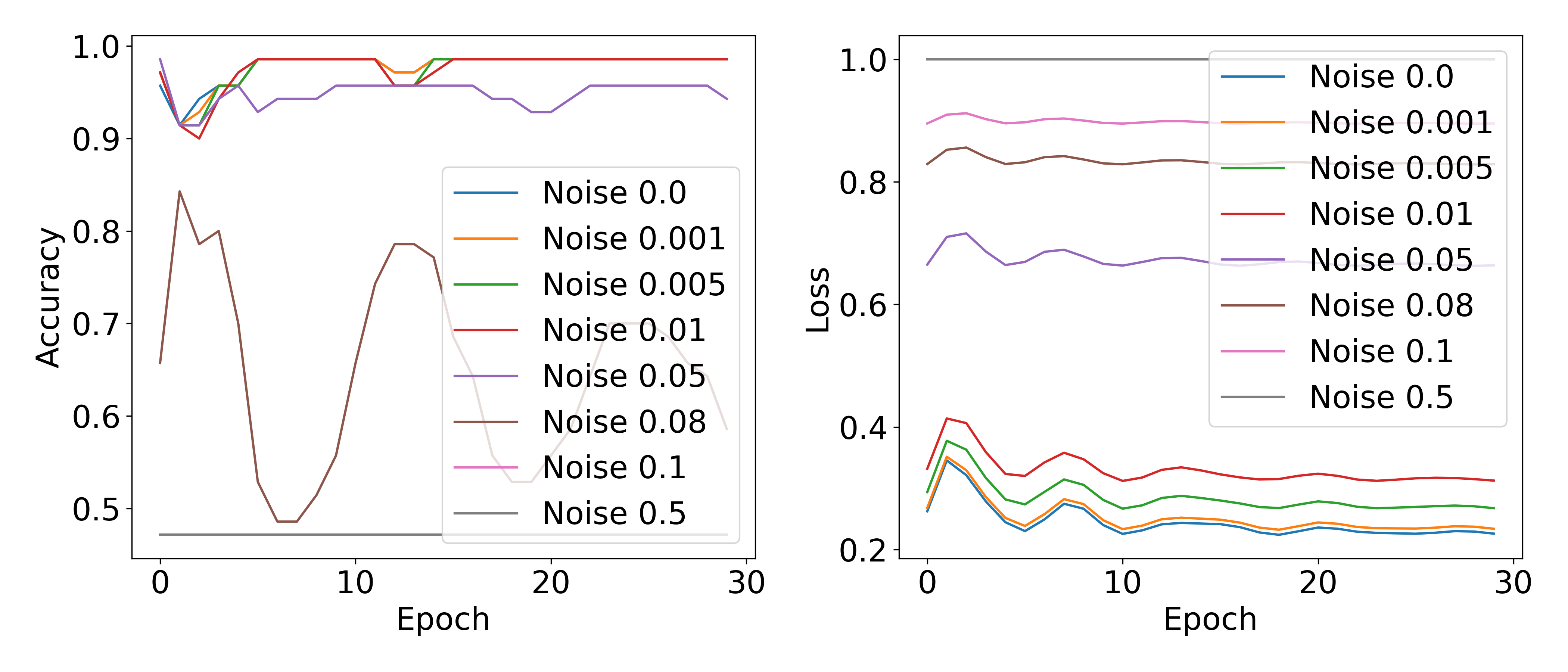} 
\caption{}
\label{fig:acc15}
\end{subfigure}
\hfill
\begin{subfigure}[b]{0.48\textwidth}
\includegraphics[width=0.9\linewidth]{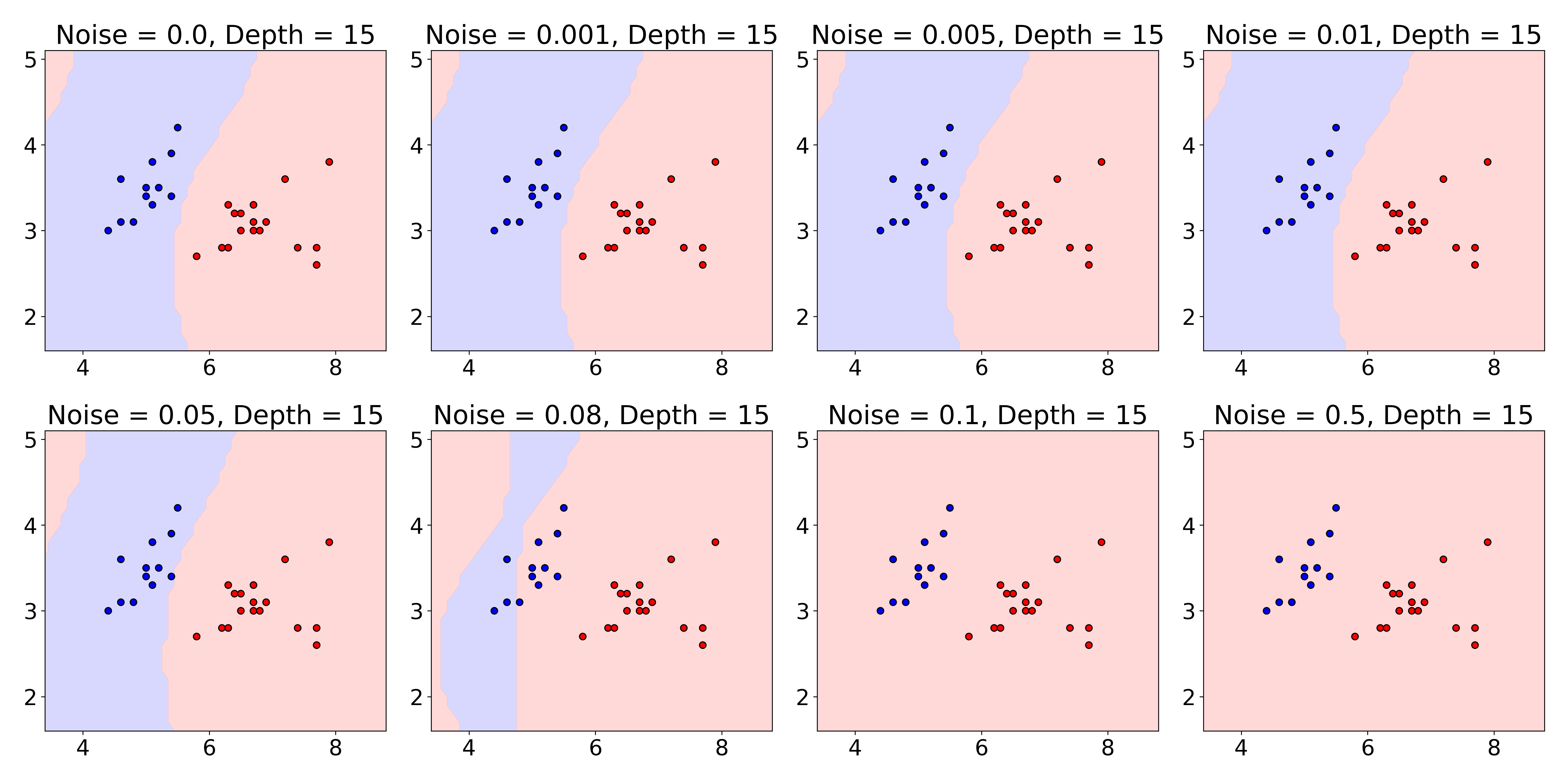}
\caption{}
\label{fig:dec15}
\end{subfigure}
\caption{Experimental results for decision boundary evolution presented in the right column and training dynamics in the left column for a QML model on the Iris dataset, with varied noise levels ($p$) and depolarization channel applied up to ($m$) times. The decision boundaries are plotted for depths of $1, 3, 5, 10,$ and $15$, at noise levels ranging from$ 0.0$ to $0.5.$ The results across rows are presented in chronological order in circuit depth. Accuracy and loss graphs display the model's performance over $30$ epochs that highlight the impact of noise rate and circuit depth on learning efficacy.}
\label{fig:all}

\end{figure*}


We trained the model for various values of m=[$1,3,5,10,15$]
and  p=[0$.0, 0.001, 0.005, 0.01, 0.05, 0.08, 0.1,0.5$] 
The initial learnable parameters were chosen randomly. We used the Adam optimizer with a $0.1$ step size, square loss as a loss function, and a parameter-shift rule to compute the gradients. We trained the model for $30$ epochs for each combination of $p$ and $m$.

The pictorial representation of the depolarization noise at each layer of the circuit and its influence on the QML model training is presented in Fig.~\ref{fig:all}. The decision boundaries presented in the right column of the figure illustrate the model's behavior across varying levels of noise and circuit depths.  We observe that with an increase in circuit depth, while the expressivity of the model may enhance, there is a concurrent increase in its vulnerability to noise, which adversely affects the quality of the decision boundary. Notably, the model's performance appears to be robust for circuit depths ranging from $ 1-5$. The model is capable of creating a decision boundary at a depth of $5$, even when the depolarization rate is $0.1$ or at a depth of $15$ for a $0.05$ error rate. This consistent performance of the model at lower noise levels across all circuit depths indicates that QML models can be robust to noise up to a certain threshold.

However, we observe that the model has difficulty making accurate predictions, even with one or three trainable gates when the noise level is $0.5$ or higher. This indicates a possible threshold for both circuit depth and noise level that maximizes QML model performance.  In our experiments, reaching a plateau at a depth of five suggests that the model's capacity for feature representation may be sufficiently saturated by a depth of five. Beyond this point, we observed a decline in precision, F1 score, and accuracy. Hence, the results indicate that while QML models exhibit robustness in lower noise environments and at shallower circuit depths, their performance diminishes with increased circuit complexity and higher noise levels.

In the next section, we discuss the advantages and limitations of a modified depolarization channel and discuss the trade-off between model complexity and noise resilience.

\section{Discussions}\label{sec:discussion}
We argue that~(\ref{eq:modifiedDepolarize}) and (\ref{eq:expvalofmodifiedm}) provide the computation advantage for simulating depolarization noise. As long as $p$ is small, similar to (\ref{eq:expvalofmodifiedm}), we can define the density matrix of a system under depolarization channel up to $m$ times for (\ref{eq:generalDepolarize}) as follows:
\begin{equation}\label{eq:generaluptom}
    \rho'_m = (1 - mp) \rho + \frac{mp}{3} (X\rho X + Y \rho Y + Z \rho Z) + \mathcal{O}(p^2) .
\end{equation}
The expectation value of an observable $O$ can be defined as:
\begin{multline}\label{eq:generalmexp}
    \langle O \rangle_{ \rho'_m} = \text{Tr}\{O \rho'_m\} = \text{Tr}\{O \rho \} - mp \text{Tr} \{O\rho\} + \\
    \frac{mp}{3} \left[\text{Tr}(OX\rho X) + \text{Tr}(OY\rho Y) + \text{Tr}(OZ \rho Z)\right] .
\end{multline}

We conducted a study to analyze the computational efficiencies of two depolarizing channel models applied to quantum states. These models were given by (\ref{eq:generalDepolarize}) and (\ref {eq:modifiedDepolarize}). We focused on the matrix multiplication overhead and the operational requirements, specifically the use of Pauli gates. 

Our findings suggest that the standard depolarizing channel, given by (\ref{eq:generalDepolarize}) and (\ref {eq:generaluptom}), requires a significantly higher computational investment. It requires six matrix multiplications for state evolution and an additional four for computing expectation values, as shown in (\ref{eq:generalmexp}), totaling ten multiplications. This approach also demands the use of all three Pauli gates $(X, Y, Z)$, which can introduce complexity in gate operations and potential errors in practical quantum computing environments. 

On the other hand, the modified depolarizing channel, given by (\ref{eq:modifiedDepolarize}) and (\ref {eq:modifiedDepolarizeUptoMTimes}), presents a more efficient alternative. It requires only four matrix multiplications for state evolution and two additional ones for expectation value computations, according to (\ref{eq:expvalofmodifiedm}), totaling six multiplications. The modified model also eliminates the need for direct Pauli $Y-$gate applications, thereby simplifying the operational framework. 

Our analysis highlights the modified channel's potential to reduce computational overhead and operational complexity. In practical scenarios, such as cloud-based quantum computing environments, users often face lengthy queues, leading to extended wait times, sometimes spanning hours, for executing a single operation. By reducing the matrix multiplications, our approach effectively reduces the computational load by at least 16 multiplication and 8 addition operations. Considering that a typical user might only manage to perform one operation per hour on real quantum hardware due to these queuing constraints, our method could result in substantial time savings. This makes it ideal for use in quantum algorithms where efficiency and gate operation minimization are important.

We further argue that the modified channel can be extended to train the QML model. The results on QML model behavior demonstrate the nuanced interplay between circuit depth, noise levels, and the model's performance. The results suggest that there exists a level of quantum circuit complexity where the representational power of the model is optimal. However, as we extend the circuit depth beyond this optimal point, we observe diminishing returns in model performance, highlighting a critical trade-off between the expressiveness of deeper quantum circuits and their susceptibility to noise.

\section{Conclusion}\label{sec:conclusion}
Our work presents the computational benefits of using a modified depolarizing channel, suggesting a more resource-efficient alternative to the standard model. The modified channel requires fewer matrix multiplications and simplifies operations by avoiding the use of the Pauli $Y-$gate, thereby providing a practical approach for NISQ-era quantum algorithms. When applied to the training of a QML model on the Iris dataset, this model demonstrates an optimal balance between circuit depth and noise resilience, with performance metrics indicating a peak in representational power at intermediate circuit depths. Beyond this threshold, the model's performance weakens, emphasizing the delicate interplay between circuit depth, noise rate and dataset complexity. The findings of our study not only advance the understanding of depolarizing noise in quantum systems but also guide the development of QML models, ensuring they harness the computational advantage. Our future works include extending the application of the modified depolarization channel to multi-qubit systems and exploring its effectiveness in training complex and higher dimensional datasets.

\section*{Acknowledgements}
Part of this work was performed while P.R. was funded by the National Science Foundation under Grant Nos.  CNS - 2210091, and CHE - 1905043. 

\section*{Data Availability}
Data sharing does not apply to this article as no datasets were generated or analyzed during the current study.

\bibliography{ref}

\begin{thebibliography}{10}
\providecommand{\url}[1]{#1}
\csname url@samestyle\endcsname
\providecommand{\newblock}{\relax}
\providecommand{\bibinfo}[2]{#2}
\providecommand{\BIBentrySTDinterwordspacing}{\spaceskip=0pt\relax}
\providecommand{\BIBentryALTinterwordstretchfactor}{4}
\providecommand{\BIBentryALTinterwordspacing}{\spaceskip=\fontdimen2\font plus
\BIBentryALTinterwordstretchfactor\fontdimen3\font minus \fontdimen4\font\relax}
\providecommand{\BIBforeignlanguage}[2]{{%
\expandafter\ifx\csname l@#1\endcsname\relax
\typeout{** WARNING: IEEEtran.bst: No hyphenation pattern has been}%
\typeout{** loaded for the language `#1'. Using the pattern for}%
\typeout{** the default language instead.}%
\else
\language=\csname l@#1\endcsname
\fi
#2}}
\providecommand{\BIBdecl}{\relax}
\BIBdecl

\bibitem{mitarai2018quantum}
K.~Mitarai, M.~Negoro, M.~Kitagawa, and K.~Fujii, ``Quantum circuit learning,'' \emph{Physical Review A}, vol.~98, no.~3, p. 032309, 2018.

\bibitem{havlivcek2019supervised}
V.~Havl{\'\i}{\v{c}}ek, A.~D. C{\'o}rcoles, K.~Temme, A.~W. Harrow, A.~Kandala, J.~M. Chow, and J.~M. Gambetta, ``Supervised learning with quantum-enhanced feature spaces,'' \emph{Nature}, vol. 567, no. 7747, pp. 209--212, 2019.

\bibitem{liu2021rigorous}
Y.~Liu, S.~Arunachalam, and K.~Temme, ``A rigorous and robust quantum speed-up in supervised machine learning,'' \emph{Nature Physics}, vol.~17, no.~9, pp. 1013--1017, 2021.

\bibitem{sajjan2021quantum}
M.~Sajjan, S.~H. Sureshbabu, and S.~Kais, ``Quantum machine-learning for eigenstate filtration in two-dimensional materials,'' \emph{Journal of the American Chemical Society}, vol. 143, no.~44, pp. 18\,426--18\,445, 2021.

\bibitem{cai2015entanglement}
X.-D. Cai, D.~Wu, Z.-E. Su, M.-C. Chen, X.-L. Wang, L.~Li, N.-L. Liu, C.-Y. Lu, and J.-W. Pan, ``Entanglement-based machine learning on a quantum computer,'' \emph{Physical review letters}, vol. 114, no.~11, p. 110504, 2015.

\bibitem{ciliberto2018quantum}
C.~Ciliberto, M.~Herbster, A.~D. Ialongo, M.~Pontil, A.~Rocchetto, S.~Severini, and L.~Wossnig, ``Quantum machine learning: a classical perspective,'' \emph{Proceedings of the Royal Society A: Mathematical, Physical and Engineering Sciences}, vol. 474, no. 2209, p. 20170551, 2018.

\bibitem{farhi2014quantum}
E.~Farhi, J.~Goldstone, and S.~Gutmann, ``A quantum approximate optimization algorithm,'' \emph{arXiv preprint arXiv:1411.4028}, 2014.

\bibitem{mcclean2016theory}
J.~R. McClean, J.~Romero, R.~Babbush, and A.~Aspuru-Guzik, ``The theory of variational hybrid quantum-classical algorithms,'' \emph{New Journal of Physics}, vol.~18, no.~2, p. 023023, 2016.

\bibitem{Rebentrost2016Quantum}
P.~Rebentrost, M.~Schuld, L.~Wossnig, F.~Petruccione, and S.~Lloyd, ``Quantum gradient descent and newton’s method for constrained polynomial optimization,'' \emph{New Journal of Physics}, vol.~21, 2016.

\bibitem{Bittel2021Training}
L.~Bittel and M.~Kliesch, ``Training variational quantum algorithms is np-hard.'' \emph{Physical review letters}, vol. 127 12, p. 120502, 2021.

\bibitem{Rebentrost2018Quantum}
P.~Rebentrost and S.~Lloyd, ``Quantum computational finance: quantum algorithm for portfolio optimization,'' \emph{arXiv: Quantum Physics}, 2018.

\bibitem{Broadbent2015Quantum}
A.~Broadbent and C.~Schaffner, ``Quantum cryptography beyond quantum key distribution,'' \emph{Designs, Codes, and Cryptography}, vol.~78, pp. 351 -- 382, 2015.

\bibitem{Padamvathi2016Quantum}
V.~Padamvathi, B.~Vardhan, and A.~V. Krishna, ``Quantum cryptography and quantum key distribution protocols: A survey,'' \emph{2016 IEEE 6th International Conference on Advanced Computing (IACC)}, pp. 556--562, 2016.

\bibitem{Lai2017Fast}
H.~Lai, M.~Luo, J.~Pieprzyk, J.~Zhang, L.~Pan, S.~Li, and M.~Orgun, ``Fast and simple high-capacity quantum cryptography with error detection,'' \emph{Scientific Reports}, vol.~7, 2017.

\bibitem{Pirandola2019Advances}
S.~Pirandola, U.~Andersen, L.~Banchi, M.~Berta, D.~Bunandar, R.~Colbeck, D.~Englund, T.~Gehring, C.~Lupo, C.~Ottaviani, J.~L. Pereira, M.~Razavi, J.~S. Shaari, M.~Tomamichel, V.~C. Usenko, G.~Vallone, P.~Villoresi, and P.~Wallden, ``Advances in quantum cryptography,'' \emph{arXiv: Quantum Physics}, 2019.

\bibitem{Harrow2017Quantum}
A.~Harrow and A.~Montanaro, ``Quantum computational supremacy,'' \emph{Nature}, vol. 549, pp. 203--209, 2017.

\bibitem{preskill2018quantum}
J.~Preskill, ``Quantum computing in the nisq era and beyond,'' \emph{Quantum}, vol.~2, p.~79, 2018.

\bibitem{pop00002}
\BIBentryALTinterwordspacing
Y.~Du, M.~Hsieh, T.~Liu, S.~You, and D.~Tao, ``Learnability of quantum neural networks,'' \emph{PRX Quantum}, 2021. [Online]. Available: \url{https://link.aps.org/doi/10.1103/PRXQuantum.2.040337}
\BIBentrySTDinterwordspacing

\bibitem{Cross2014Quantum}
A.~Cross, G.~Smith, and J.~Smolin, ``Quantum learning robust against noise,'' \emph{Physical Review A}, vol.~92, p. 012327, 2014.

\bibitem{pop00001}
\BIBentryALTinterwordspacing
Y.~Du, M.~Hsieh, T.~Liu, D.~Tao, and N.~Liu, ``Quantum noise protects quantum classifiers against adversaries,'' \emph{Physical Review Research}, 2021. [Online]. Available: \url{https://journals.aps.org/prresearch/abstract/10.1103/PhysRevResearch.3.023153}
\BIBentrySTDinterwordspacing

\bibitem{pop00023}
\BIBentryALTinterwordspacing
J.~Huang, Y.~Tsai, C.~Yang, C.~Su, and ..., ``Certified robustness of quantum classifiers against adversarial examples through quantum noise,'' \emph{ICASSP 2023-2023 …}, 2023, query date: 2024-01-23 15:43:50. [Online]. Available: \url{https://ieeexplore.ieee.org/iel7/10094559/10094560/10095030.pdf}
\BIBentrySTDinterwordspacing

\bibitem{west2023towards}
M.~T. West, S.-L. Tsang, J.~S. Low, C.~D. Hill, C.~Leckie, L.~C. Hollenberg, S.~M. Erfani, and M.~Usman, ``Towards quantum enhanced adversarial robustness in machine learning,'' \emph{Nature Machine Intelligence}, pp. 1--9, 2023.

\bibitem{lu2020quantum}
S.~Lu, L.-M. Duan, and D.-L. Deng, ``Quantum adversarial machine learning,'' \emph{Physical Review Research}, vol.~2, no.~3, p. 033212, 2020.

\bibitem{pop00008}
\BIBentryALTinterwordspacing
A.~Skolik, S.~Mangini, T.~Bäck, and ..., ``Robustness of quantum reinforcement learning under hardware errors,'' \emph{EPJ Quantum …}, 2023, query date: 2024-01-23 15:43:50. [Online]. Available: \url{https://epjquantumtechnology.springeropen.com/articles/10.1140/epjqt/s40507-023-00166-1}
\BIBentrySTDinterwordspacing

\bibitem{bai2021recent}
T.~Bai, J.~Luo, J.~Zhao, B.~Wen, and Q.~Wang, ``Recent advances in adversarial training for adversarial robustness,'' \emph{arXiv preprint arXiv:2102.01356}, 2021.

\bibitem{kang2019transfer}
D.~Kang, Y.~Sun, T.~Brown, D.~Hendrycks, and J.~Steinhardt, ``Transfer of adversarial robustness between perturbation types,'' \emph{arXiv preprint arXiv:1905.01034}, 2019.

\bibitem{huang2021power}
H.-Y. Huang, M.~Broughton, M.~Mohseni, R.~Babbush, S.~Boixo, H.~Neven, and J.~R. McClean, ``Power of data in quantum machine learning,'' \emph{Nature communications}, vol.~12, no.~1, p. 2631, 2021.

\bibitem{pop00003}
\BIBentryALTinterwordspacing
X.~Wang, Y.~Du, Y.~Luo, and D.~Tao, ``Towards understanding the power of quantum kernels in the nisq era,'' \emph{Quantum}, 2021. [Online]. Available: \url{https://quantum-journal.org/papers/q-2021-08-30-531/}
\BIBentrySTDinterwordspacing

\bibitem{pop00021}
\BIBentryALTinterwordspacing
T.~Piskor, J.~Reiner, S.~Zanker, N.~Vogt, M.~Marthaler, and ..., ``Using gradient-based algorithms to determine ground-state energies on a quantum computer,'' \emph{Physical Review A}, 2022, query date: 2024-01-23 15:43:50. [Online]. Available: \url{https://journals.aps.org/pra/abstract/10.1103/PhysRevA.105.062415}
\BIBentrySTDinterwordspacing

\bibitem{nielsen2010quantum}
M.~A. Nielsen and I.~L. Chuang, \emph{Quantum computation and quantum information}.\hskip 1em plus 0.5em minus 0.4em\relax Cambridge university press, 2010.

\bibitem{Wootton2012High}
J.~R. Wootton and D.~Loss, ``High threshold error correction for the surface code.'' \emph{Physical review letters}, vol. 109 16, p. 160503, 2012.

\bibitem{fowler2012surface}
A.~G. Fowler, M.~Mariantoni, J.~M. Martinis, and A.~N. Cleland, ``Surface codes: Towards practical large-scale quantum computation,'' \emph{Physical Review A}, vol.~86, no.~3, p. 032324, 2012.

\bibitem{gottesman1997stabilizer}
D.~Gottesman, \emph{Stabilizer codes and quantum error correction}.\hskip 1em plus 0.5em minus 0.4em\relax California Institute of Technology, 1997.

\bibitem{urbanek2021mitigating}
M.~Urbanek, B.~Nachman, V.~R. Pascuzzi, A.~He, C.~W. Bauer, and W.~A. de~Jong, ``Mitigating depolarizing noise on quantum computers with noise-estimation circuits,'' \emph{Physical review letters}, vol. 127, no.~27, p. 270502, 2021.

\bibitem{Cai2020Multi-exponential}
Z.~Cai, ``Multi-exponential error extrapolation and combining error mitigation techniques for nisq applications,'' \emph{npj Quantum Information}, vol.~7, pp. 1--12, 2020.

\bibitem{pop00004}
\BIBentryALTinterwordspacing
T.~Haug, C.~Self, and M.~Kim, ``Quantum machine learning of large datasets using randomized measurements,'' \emph{Machine Learning: Science and …}, 2023. [Online]. Available: \url{https://iopscience.iop.org/article/10.1088/2632-2153/acb0b4/meta}
\BIBentrySTDinterwordspacing

\end{thebibliography}

\end{document}